\documentclass[journal,12pt,draftclsnofoot,onecolumn]{IEEEtran}
\usepackage{cite}
\ifCLASSINFOpdf
\fi
\usepackage{amsmath,amsfonts,amsthm} % Math packages
\usepackage{amssymb}
\usepackage{bbm}
\usepackage{float}
\usepackage{graphicx}
\usepackage{algorithm}
\usepackage{algpseudocode}
\usepackage{subcaption}
\usepackage{siunitx}
\usepackage{tikz}
\usetikzlibrary{arrows,automata}
\usepackage{multicol}

\usepackage{pgfplots}
\usepackage{epstopdf}
\usepackage{multirow}
\usepackage[font=scriptsize,skip=5pt]{caption}
\usepackage{soul}
\usepackage{todonotes}
\usepackage{color, colortbl}
\usepackage{array}
\usetikzlibrary{patterns}
\usepackage{balance}

\usepackage{enumitem}
\usepackage{footmisc}

\usepackage{comment}
\usetikzlibrary{shapes.geometric}

\graphicspath{{images/}}
\newtheorem{theorem}{Theorem}

\newtheorem{lemma}[theorem]{Lemma}
\newtheorem{proposition}[theorem]{Proposition}

\definecolor{Gray}{gray}{0.9}
\definecolor{LightCyan}{rgb}{0.88,1,1}

\algnewcommand{\IIf}[1]{\State\algorithmicif\ #1\ \algorithmicthen}
\algnewcommand{\IElse}[2]{\State\algorithmicelse\ #2\ }
\algnewcommand{\EndIIf}{\unskip\ \algorithmicend\ \algorithmicif}

\IEEEoverridecommandlockouts
\begin{document}
\setlength{\parskip}{0em}

\title{Energy Harvesting Communications with Batteries having {Cycle} Constraints
	%Search term "full charge and discharge cycles" ieee
	%\thanks{Rajshekhar Vishweshwar Bhat and Mehul Motani are with the Department of Electrical and Computer Engineering, National University of Singapore, Singapore (e-mails: rajshekhar.bhat@u.nus.edu, motani@nus.edu.sg). Chandra R. Murthy is with the Dept. of ECE, Indian Institute of Science, Bangalore 560012, India (e-mail:
	%	cmurthy@iisc.ac.in). 	Rahul Vaze is with the School of Technology and Computer Science, Tata Institute of Fundamental Research, Mumbai 400005, India (e-mail:
	%	vaze@tcs.tifr.res.in).} 
	}
\author
{
\IEEEauthorblockN
{   Rajshekhar Vishweshwar Bhat,~\IEEEmembership{Graduate~Student~Member,~IEEE,}
    Mehul Motani,~\IEEEmembership{Fellow,~IEEE,}
    Chandra R Murthy,~\IEEEmembership{Senior~Member,~IEEE,}		
    and Rahul Vaze,~\IEEEmembership{Senior~Member,~IEEE}
}
}
\maketitle
\begin{abstract}
	Practical energy harvesting (EH) based communication systems typically use a  battery to temporarily store the harvested energy prior to its use for communication. 
{The batteries can be damaged when they are repeatedly charged (discharged) after being partially discharged (charged), overcharged or deeply discharged. 
	%The capacity can be recovered by imposing 
	This motivates the \emph{cycle constraint} which says that a battery must be charged (discharged) only after it is sufficiently discharged (charged). 
	We also assume Bernoulli energy arrivals,  and a half-duplex constraint due to which the batteries are not charged and discharged simultaneously.} 
	%In this work, we consider and compare energy harvesting communication systems under two cases: (a) the \emph{single-battery} case, in which the  transmitters are equipped with a single battery having storage capacity  of $2B$ units, and  (b) the \emph{dual-battery} case, in which the transmitters are equipped with two batteries, each having storage capacity of $B$ units.  
	%Many practical batteries, when repeatedly charged after being partially discharged and vice versa, suffer from degradation in the usable capacity. The capacity can be recovered by completely discharging the battery before charging it fully again and vice versa. Hence, in this work, we impose the following \emph{cycle} constraint: a battery must be charged (discharged) only after it is fully discharged (charged). We also note that practical batteries cannot be charged and discharged simultaneously. 
	In this context, we study EH communication systems with: (a) a \emph{single-battery}  with capacity $2B$ units and (b)  \emph{dual-batteries}, each having capacity of $B$ units. The aim is to obtain the best possible long-term average throughputs and throughput regions in point-to-point (P2P) channels  and multiple access channels (MAC), respectively. 
	For the P2P channel, we obtain an analytical optimal solution in the single-battery case, and propose optimal and suboptimal power allocation policies for the dual-battery case.
	%The proposed policies subsume non-adaptive policies, which do not exploit knowledge of the current battery states, and adaptive policies, which adapt power allocations based on battery states when a new energy arrival occurs. 
	We extend these policies to obtain achievable throughput regions in MACs by jointly allocating rates and powers. 
	From numerical simulations, we find that  the optimal throughput in the dual-battery case is significantly higher than that in the single-battery case, although the total storage capacity in both cases is $2B$ units.  
	Further, in the proposed policies, the largest throughput region in the single-battery case is contained within that of the dual-battery case. 
	%However, the optimal policy in the single-battery case is more fair than that in the dual-battery case. 
\end{abstract}
\IEEEpeerreviewmaketitle
%%%%%%%%%%%%%%%%%%%%%%%%%%%%%%%%%%%%%%%%%%%%%%%%%%
\section{Introduction}
Energy harvesting (EH) from natural and man-made sources has been envisioned as a viable technique for powering emerging low-power and energy starved communication systems, including Internet of Things (IoT) devices in fifth generation (5G) networks \cite{5G-EH}. 
Consequently, EH communications has been widely studied in the last few years \cite{Review,BhatMotaniLim,chandraiisc,Rahul_online_TIT}. 
In most of the existing studies, the proposed power allocation  policies may require batteries to undergo repeated \emph{partial} charge and discharge cycles. 
In practice, such a charging and discharging pattern affects the usable capacity of batteries. For instance, the usable capacity of many Nickel based batteries \cite{memory-effect} and some Li-ion \cite{Sasaki2013} batteries reduces at a significant rate with the number of  charge/discharge cycles,  {especially if the battery is partially charged/discharged}. 
This phenomenon, referred to as the memory effect or voltage depression, can be avoided by imposing the so-called { \emph{cycle constraint}, i.e.,  by discharging (charging) the battery to a lower (an upper) limit  before charging (discharging) it again \cite{Voltagedepression}.   The lower and upper limits prevent deep discharging and overcharging, respectively, which are known to damage batteries.}
%In this work, we do not focus on modeling the voltage depression, but rather, we charge and discharge the batteries such that the voltage depression does not occur, i.e., we apply the cycle constraint.  
Further, we also note that practical batteries cannot be charged and discharged simultaneously (half-duplex constraint). Hence, in this work, we consider the design of EH communication systems under (i) the cycle constraint, and (ii) the half-duplex constraint.  

Our goal in the work is to obtain power policies that maximize the long-term average throughputs and throughput regions in single-battery and dual-battery cases in a P2P channel and a MAC under constraints (i) and (ii).
We assume that the throughput is a concave function of the transmit power. 
Due to the above constraints, in the single-battery case, when the battery is being discharged (charged), energy harvesting (transmission) is suspended, and when the battery gets empty (full), charging (transmission) starts. Hence, if the energy from the battery is utilized aggressively, the duration of transmission will be short, implying that  the  time duration over which energy harvesting is suspended is short. However, due to the concavity,  the time-averaged throughput in an \emph{aggressive} policy could be less than that with a \emph{conservative} policy with a slower energy utilization. Similarly, in the dual-battery case, when a battery is being charged, the transmitter draws power from the other battery and the roles of the batteries are switched when the \emph{charging} battery becomes full and the other battery gets empty. In this case, an aggressive policy leads to a lower throughput,  as the working battery may get drained long before the charging battery gets full. 
However, a conservative policy may result in energy overflow, as the battery being discharged may not be empty  when the charging battery becomes full.  
Between these two extremes lies the optimal solution.

Despite receiving significant attention,  the impact of  charging and discharging dynamics on the usable capacity of batteries has been scarcely studied in the EH communications literature~\cite{Review}. 
In a P2P channel, the authors in~\cite{BatteryDegradation,Aging-TC,Aging-TGCN} consider the  interplay between the battery charging and discharging policy and the irreversible degradation (aging) of its storage capacity. This irreversible degradation is different from the voltage depression which can be avoided by applying the cycle constraint.    
In~\cite{JingYang17}, the authors indirectly control the battery degradation by constraining the number of charge and discharge cycles per unit time,  {without modeling the impact of the number of charge and discharge cycles on the battery capacity.}
In \cite{Shaviv-Ozgur}, a Bernoulli energy arrival model assumed, where, in a slot, either an energy packet with energy quantum equal to the capacity of the battery arrives, or no energy arrives. This implies, whenever a packet of energy arrives, the battery fills up completely. In this case, when a fresh energy packet arrives, the residual energy in the battery can be thought to be discarded instantaneously  before  replenishing it with the energy that arrived. Hence,  \cite{Shaviv-Ozgur} implicitly accounts for the cycle constraint since battery has unit capacity. 
In multi-user settings, as in \cite{Shaviv-Ozgur}, references  \cite{OnlineOzgur} and \cite{MAC-Pillai-letter} implicitly assume the cycle constraint and obtain near-optimal online power allocation policies with causal knowledge of harvested powers.  
Different from \cite{Shaviv-Ozgur,OnlineOzgur,MAC-Pillai-letter}, we consider a more general case where multiple energy arrivals may be needed to fill the battery. In addition to the cycle constraint, assumed in  \cite{Shaviv-Ozgur,OnlineOzgur,MAC-Pillai-letter},  we also account for the practical half-duplex constraint in the current work and   study a P2P channel and a MAC.  
%\cite{MAC-capacity} characterizes the capacity region of an EH MAC and obtains inner and outer bounds that differ by a constant gap. Further, authors in \cite{MAC-Pillai-letter}  maximize the average throughput under slot-wise coding and decoding requirements by adapting a policy for a distributed slow fading MAC  studied in \cite{MAC-Pillai-IT}. 
%In comparison to the P2P case, EH communication in multi-user settings has received relatively less attention. Achievable rate regions of an EH Gaussian MAC with infinite and finite capacity batteries have been obtained in \cite{MAC_VSharma}. Several near-optimal power allocation policies have been obtained in \cite{OnlineOzgur,MAC-ulukus,RVaze_online,MAC-Mitran,MAC-Jing-ISIT} with causal knowledge of harvested powers.  \cite{MAC-capacity} characterizes the capacity region of an EH MAC and obtains inner and outer bounds that differ by a constant gap. Further, authors in \cite{MAC-Pillai-letter}  maximize the average throughput under slot-wise coding and decoding requirements by adapting a policy for a distributed slow fading MAC  studied in \cite{MAC-Pillai-IT}.  Further, \cite{Two-Way-TGCN} considers a two-way relay channel with radio-frequency EH. 
%In this work, we study a P2P channel and a MAC. 
For clarity, we first present the P2P channel and then extend the results to a MAC. 
The main contributions of the paper are: % as follows. 
\begin{itemize}[leftmargin=*]
	\item In a P2P channel under the single-battery case, we obtain an analytical solution to the long-term average throughput maximization problem in the online case with casual knowledge of energy arrivals.  We also show that the throughput maximizing policy performs better than a competitive greedy policy from \cite{Yin-WCNC}. 
	
	%We note that, although the complexity of the optimal online policy is nearly the same as that of a competitive greedy policy, its performance is at least equal to that of the greedy policy.  
	\item Further, for a P2P channel under the dual-battery case, we first  obtain  optimal power allocations via dynamic programing and then propose non-adaptive online policies, which do not exploit knowledge of current battery states, and adaptive policies, which adapt power allocations based on battery states when a new energy arrival occurs. 
	\item For a $U$-user MAC, we then derive long-term average achievable  throughput regions in the single-battery and dual-battery cases, based on the online policies proposed for the P2P channel. 
	\item  	From numerical simulations, we show that the performance gap between the optimal policy for an ideal system equipped with infinite capacity batteries and the proposed non-adaptive policies with finite capacity batteries decays faster than the inverse of the square root of the battery capacity. Further, the largest throughput region in the single-battery case is contained within that of the dual-battery case. 
	%However, the optimal policy in the single-battery case is more fair than that in the dual-battery case.
\end{itemize}
In summary, our study finds that, under the cycle and half-duplex constraints,  the optimal performance in the dual-battery case is significantly better than that in the single-battery case, although the total storage capacity in both cases is the same.

The remainder of the paper is organized as follows. The system model is  presented in Section \ref{sec:system-model}. 
We study P2P channel under single-battery and  dual-battery  cases in  Section \ref{sec:SU-SB} and Section \ref{sec:SU-DB}, respectively. We then  consider a MAC in Section \ref{sec:MAC}.
Numerical results are presented in Section \ref{sec:numerical-results}, followed by concluding remarks in Section \ref{sec:conclusions}.

\section{System Model}\label{sec:system-model}
In this work, we consider a P2P channel and a MAC under two cases: (a) the \emph{single-battery} case, in which the  transmitters are equipped with a single battery having storage capacity  of $2B$ units, and  (b) the \emph{dual-battery} case, in which the transmitters are equipped with two batteries, each having storage capacity of $B$ units, as shown in Fig. \ref{fig:single battey} and Fig. \ref{fig:two battery}, respectively. {In both the cases, the transmitters are powered from EH sources. Further, the receiver is connected to the mains, and hence the receiver can always remain ON.
Now, based on our discussion in the introduction, we impose the half-duplex constraint, and the cycle constraint due to which the battery must be discharged (charged) to a lower limit, $C_{\rm min}$  (an upper limit, $C_{\rm max}$) before charging (discharging) it again, with $C_{\rm min}< C_{\rm max}$. In this work, without loss of generality, we assume $C_{\rm min}=0$ for both the single-battery and dual-battery cases, and $C_{\rm max}=2B$ and $C_{\rm max}=B$ for the single-battery and dual battery cases, respectively.}
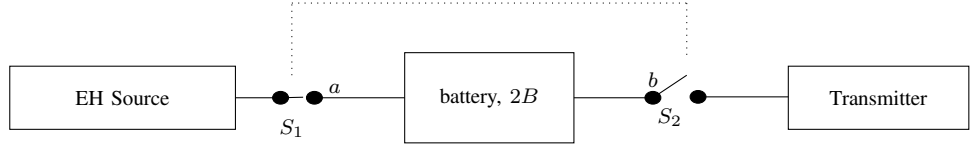
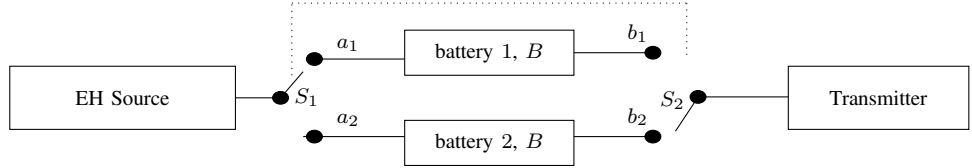
\begin{figure*}[t]
	\centering
	\begin{subfigure}{0.75\textwidth}
		\begin{tikzpicture} [scale=3,yscale=0.8]
		\draw (0,-0.05) rectangle (1,.3) node [align=center,pos=0.5] {\scriptsize EH Source};  % EH source
		\draw  (1,0.125) -- (1.20,0.125);
		\filldraw [black] (1.20,0.125) circle (1pt);
		\filldraw [black] (1.35,0.125) circle (1pt) node at ( 1.44, 0.175) {\scriptsize $a$} ;
		\draw (1.75,-0.125) rectangle (2.5,0.375) node [pos=0.5]  {\scriptsize battery, $2B$};
		\draw (1.35,0.125) -- (1.75,0.125);
		\draw (2.5,0.125) -- (2.85,0.125) ;
		\filldraw [black] (2.85,0.125) circle (1pt)  node [above] {\scriptsize $b$}  ;
		\filldraw [black] (3.05,0.13) circle (1pt);
		\draw (3.05,0.13) -- (3.45,0.13);
		\draw (3.45,-0.05) rectangle (4.25,.3) node [align=center,pos=0.5]  {\scriptsize  Transmitter};
		\draw (1.2,0.125) -- (1.30,.128) node [midway,below,yshift=-0.15cm] {\scriptsize $S_1$}; 
		\draw (2.85,.125) -- (3,0.25) node [midway,below,yshift=-0.15cm] {\scriptsize $S_2$} ; 
		\draw [dotted] (1.25,.25) -- (1.25,0.65);
		\draw [dotted] (1.25,.65) -- (3,0.65);			
		\draw [dotted] (3,.65) -- (3,0.35);
		\end{tikzpicture}
		\caption{In the single-battery case, at any instant of time, either the EH source is connected to $a$ or the transmitter is connected to $b$.}
		\label{fig:single battey}
	\end{subfigure} 
	\\
	\begin{subfigure}{0.75\textwidth}
		\begin{tikzpicture} [scale=3,yscale=0.8]
		\draw (0,-0.05) rectangle (1,.3) node [align=center,pos=0.5] {\scriptsize EH Source};  % EH source
		\draw  (1,0.125) -- (1.20,0.125);
		\filldraw [black] (1.20,0.125) circle (1pt);
		\filldraw [black] (1.35,0.34) circle (1pt) node at ( 1.5, 0.42) {\scriptsize $a_1$} ;
		\draw (1.35,0.34) -- (1.75,0.34);
		\filldraw [black] (1.35,-0.09) circle (1pt) node at (1.5, 0) {\scriptsize $a_2$} ;
		\draw (1.30,-0.09) -- (1.75,-0.09);
		
		\draw (1.75,-0.25) rectangle (2.5,0) node [pos=0.5]  {\scriptsize battery $2$, $B$ };
		\filldraw [black] (2.85,-0.09) circle (1pt)  node [above,xshift=-0.2cm] {\scriptsize $b_2$}  ;
		\draw (2.5,-0.09) -- (2.85,-.09) ;		
		
		\draw (1.75, .25) rectangle (2.5,0.5) node [pos=0.5] {\scriptsize battery $1$, $B$};
		\filldraw [black] (2.85,0.375) circle (1pt)  node [above,xshift=-0.2cm] {\scriptsize $b_1$}  ;
		\draw (2.5,0.375) -- (2.85,0.375) ;

		\filldraw [black] (3.05,0.13) circle (1pt);
		\draw (3.05,0.13) -- (3.45,0.13);
		\draw (3.45,-0.05) rectangle (4.25,.3) node [align=center,pos=0.5] at (3.85,0.13)  {\scriptsize  Transmitter};
		\draw (1.2,0.125) -- (1.30,.27) node [midway,below,xshift=0.2cm,yshift=0.1cm] {\scriptsize $S_1$}; 
		\draw (2.95,-0.06) -- (3.04,0.11) node [left] {\scriptsize $S_2$} ; 
		\draw [dotted] (1.25,.25) -- (1.25,0.65);
		\draw [dotted] (1.25,.65) -- (3,0.65);			
		\draw [dotted] (3,.65) -- (3,0.35);
		\end{tikzpicture}
		\caption{In the dual-battery case, at any instant of time, the EH source and the transmitter will be connected to $a_1$ ($a_2$) and $b_2$ ($b_1$), respectively.}
		\label{fig:two battery}
	\end{subfigure}
	\caption{We consider EH users equipped with (a) single battery with  capacity of $2B$ units and (b) two batteries, each with  capacity  of $B$ units.}
	\label{fig:system model}
\end{figure*}
The harvested energy is stored in a battery before using it for  transmission.  
In the single-battery case, when the battery is being charged, no transmission is carried out. When transmission occurs, the energy harvesting is suspended, as batteries cannot be charged and discharged simultaneously. Hence, in Fig. \ref{fig:single battey}, when $S_1$ ($S_2$) is closed, $S_2$ ($S_1$) will be open. 
On the other hand, in the dual-battery case, when a battery is being charged from the EH source, the transmitter draws power from other battery. Hence, at any point in time, the EH source and the transmitter will be connected to $a_1$ ($a_2$) and $b_2$ ($b_1$), respectively, as shown in Fig. \ref{fig:two battery}.   
In the sequel, we present the energy arrival model, evolution of battery states and the throughput model adopted in the work.

\subsection{Energy Arrival Model}
We consider a time-slotted system with unit slot length. Energy arrivals in each slot are independent and identically distributed (i.i.d.) Bernoulli random variables with parameter $p$, where $E_H$ units of energy is harvested per arrival, i.e.,  the amount of energy harvested in slot $i$, 
\begin{align}\label{eq:EA}
E_i&=\left\{ 
\begin{array}{l l}
E_H  & \text{w.p. $p$,}\\
0  &   \text{w.p. $1-p$,}\\
\end{array} \right.
\end{align} 
where w.p. stands for ``with probability''. The harvested energy arrives continuously at a constant rate in a slot such that the total energy accumulated at the end of the slot is $E_H$ units.   We also assume $B$ and $E_H$ are related as $B/E_H=r$ for some integer $r\geq 1$ and  note that the average EH rate, $\mu =  pE_H$.

%We consider a time-slotted system with unit slot length. Energy arrivals in each slot are independent and identically distributed (i.i.d.) Bernoulli random variables with parameter $p$, where $B/r,r\in \mathbb{Z}^+$, units of energy is harvested per arrival. That is, the amount of energy harvested in slot $i$ is, 
%\begin{align}\label{eq:EA}
%E_i&=\left\{ 
%\begin{array}{l l}
%\frac{B}{r}  & \text{w.p. $p$,}\\
%0  &   \text{w.p. $1-p$,}\\
%\end{array} \right.
%\end{align} 
%where w.p. stands for \emph{with probability}. The harvested energy arrives continuously at a constant rate in a slot such that the total energy accumulated at the end of the slot is ${B}/{r}$ units. %and it is 

\subsection{Evolution of Battery States}\label{sec:bat_ev}
%We now describe the evolution of battery states. 
In the single-battery case, let $B_i$ be the amount of energy stored in the battery at the start of slot $i$. Then, in the charging phase, the time interval over which the battery is charged,    we have, 
\begin{align}
B_{i+1}=\min(B_i+E_i,2B), \;\; i=1,2,\ldots. 
\end{align}
Further, assuming the transmit power in slot $i$ is $P_i$, in the discharging (transmission) phase, the time interval over which the battery is discharged for transmission, we have, 
\begin{align}
B_{i+1}=\max(B_i-P_i,0), \;\; i=1,2,\ldots, 
\end{align}
where charging phase starts when $B_i=0$ and the transmission phase starts when $B_i=2B$. 
Similarly, in the dual-battery case, let $B^w_i$ and $B^c_i$ denote the amount of energy stored in the working and the charging battery at the start of slot $i$, respectively. We assume $B^w_1=B$ and $B^c_1=0$. Then,  the charging and working batteries evolve as follows.  
\begin{align}
B^c_{i+1}=\min(B^c_i+E_i,B), \;\; i=1,2,\ldots,\\
B^w_{i+1}=\max(B^w_i-P_i,0), \;\; i=1,2,\ldots. 
\end{align}
The min and max in the above equations capture the facts that the battery energy cannot exceed its capacity or become negative, respectively.

\subsection{Throughput and an Upper-Bound}
The communication is over an additive white Gaussian noise (AWGN) channel with unit noise power. We assume the throughput with average received signal-to-noise ratio of $P$ \si{\watt} in a slot is given by $R=\frac{1}{2}\log(1+P)$ bits-per-second (bps). All the logarithms are to the base 2. Now, the long-term average throughput is  defined as follows. 
\begin{align}\label{eq:avgT}
T\triangleq \lim_{k\rightarrow \infty}\frac{1}{k}\mathbb{E}\left[\sum_{i=1}^{k}\frac{1}{2}\log(1+P_i)\right]
\end{align}
where the expectation is over all the possible sequences of energy arrivals.  
Our objective is to find the optimal power policy, $P_1,P_2,\ldots$, that maximizes the long-term average throughput, $T$ in \eqref{eq:avgT}.   
In order to benchmark the proposed power allocation policies, we now present an upper-bound on $T$.

{When the cycle and half-duplex constraints are not present, and the nodes are equipped with infinite capacity batteries}, it is optimal to utilize $\epsilon$ smaller amount of power per slot on average, compared to the average harvesting rate, where $\epsilon > 0$ can be arbitrarily small. Further, by the concavity of the rate function, it is optimal to utilize an equal amount of energy $\mu = pE_H=Bp/r$ in every slot, resulting in the long-term average throughput \cite{unconstrained_capacity,Koksal,MAC_VSharma},
\begin{align}\label{eq:ub}
T_{\rm ub}=\frac{1}{2}\log(1+\mu).
\end{align}
The subscript `${\rm ub}$' indicates that this performance is an upper bound on the performance with any other constraints.

\section{P2P Channel: Single-Battery Case}\label{sec:SU-SB}
In this section, we consider a P2P channel under the single-battery case. We first formulate the long-term average throughput maximization problem in the online case and then obtain a closed-form optimal online power allocation policy. We then note that the  optimal online policy allocates transmit energy based on the average energy harvesting rate and  compare the performance of the optimal online policy with a competitive greedy policy from \cite{Yin-WCNC}, which allocates energy based on the instantaneous energy arrival rate in a slot.

\subsection{Problem Formulation}
Let the number of slots required to fill the battery be $L$. Note that $L$ is a random variable. The number of energy arrivals required to fill the battery of capacity $2B$ units is $2r$. Now, since inter-arrival times of energy  is a sequence of geometrically distributed i.i.d random variables with mean $1/p$, we have, $\mathbb{E}[L]=2r/p=2B/(pE_H)$. 
To obtain the optimal long-term average throughput, we assume that the initial energy stored in the battery is $2B$ units.  We then discharge and charge the battery subject to the cycle constraint. Let $P_1,\ldots,P_N$ be the transmit powers over the first $N$ slots such that $\sum_{i=1}^{N}P_i=2B$.   That is, the entire energy in the battery is consumed in first $N$ slots, where $N$ is a variable to be designed. 
Starting from $(N+1)^{\text{th}}$ slot, the battery is charged for $L$ slots to fill it, over which no transmission occurs.  Clearly, $(N+L+1)^{\text{th}}$ slot is a renewal instant. The total reward and the length of the renewal period are $R_i=\sum_{i=1}^{N}\frac{1}{2}\log(1+P_i)$ bps and  $(N+L)$ slots, respectively. Hence, by renewal-reward theorem \cite{Gallager}, from \eqref{eq:avgT}, the long-term average throughput in the single-battery case is equal to,  
\begin{align}\label{eq:SB-1}
T_{\rm SB}(N, P_1,\ldots,P_N)=\frac{\mathbb{E} \left[\sum_{i=1}^{N}R_i\right]}{\mathbb{E}[N+L]}=\frac{\sum_{i=1}^{N}\frac{1}{2}\log(1+P_i)}{N+\frac{2B}{pE_H}},
\end{align}
where the $N$ in the denominator accounts for the penalty due to not storing energy while transmitting.  
Now, to maximize the long-term average throughput, we need to solve the following optimization problem. 
\begin{subequations}\label{eq:P0}
	\begin{align}
	\underset{N,P_1,\ldots,P_N}{\text{maximize}} &\;\; T_{\rm SB}(N, P_1,\ldots,P_N),\;\;\;&&\\
	\text{subject to} &\;\; \sum_{i=1}^{N}P_i= 2B,\; P_i\geq 0, \; N\in \{1,2,\ldots\}, \;i\in \{1,\ldots,N\}.  
	\end{align}
\end{subequations}
We now solve \eqref{eq:P0} in the following under the online case.

\subsection{Optimal Online Solution}\label{sec:SB_opt}
We first note that, for any $N$, due to concavity of the logarithmic function, the numerator is maximized when the transmit powers in all the slots are the same, i.e., when  $P=P_1=P_2=\ldots=P_N$, where $P\triangleq2B/N$. %We assume the user transmits at rate $R=(1/2)\log(1+P)$. 
Hence, \eqref{eq:P0} can be reformulated as the following problem. 
	\begin{align}\label{eq:P2}
	 \underset{N\in \mathbb{Z}^+}{\text{maximize}} &\;\; \frac{\frac{N}{2}\log\left(1+\frac{2B}{N}\right)}{N+\frac{2B}{pE_H}}.\;\;\;&&
	\end{align}
Note that \eqref{eq:P2} is an integer program. To solve it, we first allow $N$ to take any positive real values and solve it. We then optimally round the solution to satisfy the integer constraint of the original problem. 
\subsubsection{Relaxation}
In the relaxed problem, we denote the transmit power and number of transmission slots by $\tilde{P}$ and $\tilde{N}$, respectively.  
%Firstly, we note that all the energy in the battery must be used for the transmission. 
Hence, we have, $\tilde{N}=2B/\tilde{P}$ and the objective function in  \eqref{eq:P2}  becomes  $\mu\log(1+\tilde{P})/({2(\mu+\tilde{P})})$, where the average EH rate, $\mu=pE_H$ under the assumption that $r\geq 1$, i.e.,   $B\geq E_H$. 
Therefore, the optimal transmit power can be obtained by solving,  
\begin{align}\label{eq:SB_power}
%\max_{\tilde{P}} \frac{B\log(1+\tilde{P})}{2(B+\tilde{P}rp^{-1})}= 
 \max_{\tilde{P}\in \mathbb{R}^+} \frac{\mu\log(1+\tilde{P})}{2(\mu+\tilde{P})}. 
\end{align} 
We  present the optimal solution to \eqref{eq:SB_power} in the following lemma. 
\begin{lemma}\label{thm:SB}
	The optimal solution to \eqref{eq:SB_power} is given by
	\begin{align}\label{eq:SB opt_tx_power}
	\tilde{P}^*%&=\frac{\mu-1}{W_0(\exp(-1)(\mu - 1))}-1\nonumber\\
	&=\exp(1)\exp\left(W_0\left(\exp(-1)(\mu-1)\right)\right)-1, 
	\end{align}
	and 
	the optimal long-term average throughput is given by 
	\begin{align}\label{eq:SB-real}
	\tilde{T}_{\rm SB}=\frac{\mu}{2\ln 2\exp(1)\exp\left(W_0\left(\exp(-1)(\mu-1)\right)\right)},
	\end{align}
	where $W_0(\cdot)$ is the principal branch of the Lambert W function. 
\end{lemma}
\begin{proof}
	See Appendix A.  
\end{proof}
\subsubsection{Rounding}
From Lemma \ref{thm:SB}, clearly, $\tilde{N}^*=2B/\tilde{P}^*$. Noting that $\tilde{N}^*$ may not be an integer, we now obtain the optimal solution to \eqref{eq:P2}, with the integer constraint on $N$. 
We note that the objective function in \eqref{eq:P2} is unimodal as it is an increasing  function over $N\in [0,\tilde{N}^*)$ and  a decreasing function over $N\in [\tilde{N}^*,\infty)$. Hence, we choose the optimal $N^*$ that satisfies the integer constraint as follows:  
\begin{align}\label{eq:intN}
N^*=  \underset{N\in\{\lfloor \tilde{N}^*\rfloor,\lceil\tilde{N}^*\rceil\}}{\arg\max}
\frac{\frac{N}{2}\log(1+\frac{2B}{N})}{N+\frac{2B}{\mu}}.
\end{align}
where $\lfloor x\rfloor$ is the greatest integer smaller than $x$ and  $\lceil x\rceil$ is the smallest integer greater than $x$. 
Hence, the maximum long-term average throughput is given by
\begin{align}\label{eq:SB-int}
{T}_{\rm SB}=\frac{\frac{N^*}{2}\log(1+\frac{2B}{N^*})}{N^*+\frac{2B}{\mu}},
\end{align}
where $N^*$ is given by \eqref{eq:intN}.

%\begin{itemize}[leftmargin=*]
%	\item 
	From Lemma \ref{thm:SB}, which gives the optimal solution to the relaxed long-term average throughput maximization problem when $B\geq E_H$, we note that the optimal power and throughput depend only on $\mu$, i.e., for any $B$ such that $B\geq E_H$, the optimal power and throughput remain the same. {In other words, the performance with an  infinite capacity  battery can be obtained by using a single battery with capacity of $2B=2E_H$ units, in the relaxed case.} %However, in the presence of integer constraints, it is difficult to comment on the dependency of the optimal throughput on the battery capacity. 
		Further, in the presence of integer constraints, from \eqref{eq:intN} and  \eqref{eq:SB-int}, we note that the optimal performance depends both on the mean value of the harvested energy and the battery capacity. This implies that burstiness of the energy arrivals does not impact the optimal performance. 
		%This is because, the cycle constrain enforces energy averaging, i.e., Further, in a later section, we will see that from numerical simulations, the rounding has negligible impact on the performance

%\item Further, we note the optimal transmit power does not depend on the state of the battery. However, switching between charging and transmission phases requires knowledge of slots in which the battery becomes full and empty, which can be achieved without much difficulty in practical systems.
%\end{itemize}
%{From Lemma \ref{thm:SB}, it is interesting to note that the optimal power and throughput in the relaxed problem in \eqref{eq:SB_power} depend only on the average EH rate, $\mu$. When the battery capacity is greater than or equal to the amount of energy arrival in a slot, i.e., when $B\geq E_H$, clearly, $\mu=pE_H$, for any $B\geq E_H$. } 
%We now recognize that the complexity our policy is nearly the same as that in a greedy policy studied in \cite{Yin-WCNC}, which does not require any knowledge of the battery states. 
%Hence, we now compare our policy with the greedy policy in the following. 

We now compare the above policy with a competitive greedy policy in the following.

\subsection{Comparison with the Competitive Greedy Policy Studied in \cite{Yin-WCNC}}
The authors in \cite{Yin-WCNC} consider a slotted system with unit slot length. {Further,  they impose only the half-duplex constraint, and  not the cycle constraint.}  
In a slot, the harvested energy is stored for $(1-\tau)$ seconds and the transmission occurs over the remaining $\tau$ seconds. The long-term average throughput maximization problem considered in \cite{Yin-WCNC} is given by
	\begin{align}\label{eq:P22}
	\underset{0\leq \tau\leq 1}{\text{maximize}} &\;\; \frac{p\tau}{2} \log\left(1+\frac{(1-\tau)E_H}{\tau}\right),
	\end{align}
where $E_H$ is the total harvested energy in a slot.  The transmit power is given by $P=(1-\tau)E_H/\tau$. We now give the optimal solution to \eqref{eq:P22} in the following lemma. 
\begin{lemma}\label{thm:SBYin}
	The optimal transmit power for the optimization problem in \eqref{eq:P22} is given by
	\begin{align}\label{eq:SB opt_tx_powerYin}
	P^*=\exp(1)\exp\left(W_0\left(\exp(-1)(E_H-1)\right)\right)-1, 
	\end{align}
the optimal transmit duration $\tau^*=E_H/(P^*+E_H)$, and 
	 the maximum long-term average throughput, 
	\begin{align}\label{eq:SB-realYin}
	{T}_{\rm greedy}=\frac{pE_H}{2\ln 2\exp(1)\exp\left(W_0\left(\exp(-1)(E_H-1)\right)\right)}.
	\end{align}
\end{lemma}
\begin{proof}
The result follows from  \cite{Yin-WCNC}, with appropriate mapping of variables followed by simplification along the lines in the proof of Lemma \ref{thm:SB}. 
\end{proof}
 From the above lemma, we note the optimal performance in the greedy policy depends on both $p$ and $E_H$, unlike in the earlier case where the optimal performance depended only on $\mu$, under the assumption that $B\geq E_H$.  Noting that the denominator in \eqref{eq:SB-realYin} is an increasing function of $E_H$, for a given $\mu=pE_H$, ${T}_{\rm greedy}$ decreases as $E_H$ increases and $p$ decreases.

In the following proposition, we compare the long-term average throughputs in the optimal online policy and the greedy policy.  
Since $\tau$ can take real values, it is meaningful to compare $T_{\rm greedy}$ with $\tilde{T}_{\rm SB}$ in \eqref{eq:SB-real}. 
\begin{proposition}\label{prop:greedy_vs_SB}
$T_{\rm greedy}< \tilde{T}_{\rm SB}$, for any $p<1$ and $E_H>0$. 
\end{proposition}
\begin{proof}
By definition, note that $E_H> \mu$ for any $p<1$ and $E_H>0$. Since $W_0(x)$ is a strictly increasing function on $[0,\infty)$, we have, $W_0(\exp(-1)(\mu-1))< W_0(\exp(-1)(E_H-1))$. The result follows by noting that $\exp(\cdot)$ is a strictly increasing function. 
%When $E_H>0$, $E_H=\mu$ holds only when $p=1$. Since $W_0(\cdot)$ and $\exp(\cdot)$ are strictly monotonic functions, $T_{\rm greedy}= \tilde{T}_{\rm SB}$ holds only when $p=1$. 
\end{proof}
From Proposition \ref{prop:greedy_vs_SB}, we see that 
the optimal online policy in Lemma \ref{thm:SB}  outperforms the competitive greedy policy in Lemma \ref{thm:SBYin}. We note that both the above polices can be implemented easily in practical systems. 

We conclude the section by noting that the main disadvantage of the single-battery case is that energy harvesting is suspended during transmissions. This can be avoided by using an additional battery which charges up while the working battery is being discharged. Hence, in the following section, we study the two-battery case.

\section{P2P Channel: Dual-Battery Case}\label{sec:SU-DB}
We now consider the P2P channel under the dual-battery case. 
In the sequel, we formulate the long-term average throughput maximization problem, and derive the optimal solution in the online case. We then propose non-adaptive online policies, which do not need knowledge of the current state of the batteries, based on which we propose policies that  adapt the power allocation based on battery state when each new energy arrival occurs. %For completeness of benchmarking, we  also present the optimal offline policy in Appendix B.
  
\subsection{Problem Formulation}
Due to the cycle constraint, which says that the working (charging) battery must be discharged (charged) completely before switching the roles of the batteries, the switching instants are  renewal instants of the battery state processes.  Let $L_k$ denote the length (in slots) of the $k^{\text{th}}$ renewal period. Let the number of slots required for the charging battery to accumulate $B$ units  of energy in the $k^{\text{th}}$ renewal period be $C_k$,  which is less than or equal to $L_k$ by our design.  
Since the system is reset after a renewal instant and because the energy arrivals and allocations are independent across renewal intervals,  $L_1,L_2,\ldots$ form a sequence of i.i.d. random variables. The $k^{\text{th}}$ renewal occurs in the $S_k^{\text{th}}$ slot, where,  
\begin{align}
S_{k+1}=S_{k}+L_k,\qquad k=1,2,\ldots,
\end{align}
where we define $S_1\triangleq 1$.
Further, due to the constraint described above, we must have, 
\begin{align}\label{eq:constraint}
B^{w}_{S_k}=B,\;\;\text{and}\;\;  B^{c}_{S_k}=0,\qquad \forall\; k=1,2,\ldots
\end{align}
Now, to maximize the long-term average throughput defined in \eqref{eq:avgT},  we need to solve the following optimization problem. 
\begin{subequations}\label{eq:p2p-1}
	\begin{align}
	\underset{P_j,\;j=1,2,\ldots}{\text{maximize}} &\;\;\;T &&\\
	\text{subject to} &\;\; \eqref{eq:constraint},  \sum_{k=S_{i-1}}^{S_i-1}P_k\leq B, \; P_j\geq 0,\; i\in \{2,3,\ldots\},\; j\in \{1,2,\ldots\},  \label{eq:energy_neutrality}
	\end{align}
\end{subequations}
where the constraint $\sum_{k=S_{i-1}}^{S_i-1}P_k\leq B$ is because the maximum total amount of energy that can be consumed in a renewal period is $B$. 

Before we proceed to obtain the optimal solutions under various cases, we note the following. 
Firstly, by renewal-reward theorem \cite{Gallager}, from \eqref{eq:avgT}, the long-term average throughput is equal to,  
\begin{align}\label{eq:rate}
T=\frac{\mathbb{E}\left[\sum_{i=1}^{L}\frac{1}{2}\log(1+P_i)\right]  }{\mathbb{E}[L]}, 
\end{align}
where $L$ is the length of the renewal period, whose distribution is identical to  $L_1, L_2,\ldots$.  
The expectation is with respect to the random variable $L$. 
Secondly,  the random variable $C$, which is the number of slots required to accumulate $B$ units of energy, has the negative binomial distribution given by,
\begin{align}\label{eq:neg_bino}
\mathrm{Prob}(C=n)\triangleq q_n={{n-1}\choose{n-r}}p^{r}(1-p)^{n-r},
\end{align}
for $n \in \{r, r+1,\ldots\}$ and $\mathbb{E}(C)=r/p$. Further, the cumulative density function (CDF) is given by, $F_i(r,p)=\sum_{n=1}^{i}q_n$ and the complementary CDF, $\bar{F_i}(r,p)=1-F_i(r,p)=\sum_{n=i+1}^{\infty}q_n$.

\subsection{Optimal Offline Policy}
In the offline policy, we assume the number of slots required to completely charge the battery is known at the start of the current renewal instant, i.e., the realization of $C_{k+1}$ is known at start of slot $S_k$ for $k\in \{1,2,\ldots\}$.  
Due to the concavity of $\frac{1}{2}\log(1+P)$ in $P$, it is optimal to transmit with a constant power over the $C_{k+1}$ slots. Hence, when $C_{k+1}$ is known, we can discharge the working battery such that it gets emptied in the same slot in which  the charging battery fills up. Hence, we have $L_{k+1}=C_{k+1}$ and we transmit at the constant power of $B/L_{k+1}$  in the $k^{\text{th}}$ renewal. Hence, from \eqref{eq:rate} and \eqref{eq:neg_bino}, the optimal long-term average throughput is given by
\begin{align}
T_{\rm off}&=\frac{p}{r}\sum_{n=r}^{\infty}\frac{nq_n}{2}\log\left(1+\frac{B}{n}\right). 
\end{align}
It is easy to numerically compute the value of $T_{\rm off}$.

\subsection{Optimal Online Policy}
We now consider the online case with only causal knowledge of renewal instants. Since the precise time when the charging battery will get full is unknown, power is allocated based only on the distribution of energy arrivals. 
To obtain the optimal online policy, we adopt a dynamic programming framework. We now define the relevant quantities.
\subsubsection{State space}
The state of the system is defined by 3-tuples $s\triangleq (b_1,b_2,\alpha)$, where $b_1$ and $b_2$ are the amounts of energy stored in the first and the second battery, respectively, at the start of a slot and $\alpha$ is an indicator variable defined as follows.  
\begin{equation}
\alpha= \left\{ \,
\begin{IEEEeqnarraybox}[][c]{l?s}
\IEEEstrut
1  & if the first battery is the working battery, \\
0   & if the first battery is the charging battery. 
\IEEEstrut
\end{IEEEeqnarraybox}
\right.
\label{eq:alpha}
\end{equation}
Thus, when $\alpha=1\;(\alpha=0)$,  the second battery is the charging (working) battery. The state~space,  
\begin{align}
\mathcal{S}=\{(b_1,b_2,\alpha): 0\leq b_1,b_2\leq B,\alpha\in\{0,1\} \}
\end{align}
\subsubsection{Action space and reward}
%The transmit power, $P$ defines the action that the system can take. 
The action space the system can take in  state  $s=(b_1,b_2,\alpha)\in \mathcal{S}$, 
\begin{align}\label{eq:action space}
\mathcal{A}(s)=\{P: 0\leq P\leq b_1\alpha+b_2(1-\alpha)\}
\end{align} 
The constraint on $P$ in \eqref{eq:action space} is due to the fact that $P$ units are drawn from the working battery. We recall that the reward when the system takes action $P\in\mathcal{A}(s)$ in state  $s\in \mathcal{S}$ is given by $\frac{1}{2}\log(1+P)$. 
\subsubsection{State transition probability matrix} 
The next state $s'\triangleq (b_1',b_2',\alpha')\in \mathcal{S}$ when the system takes action $P<b_1\alpha+b_2(1-\alpha)$ in  state $s=(b_1,b_2,\alpha)\in \mathcal{S}$ such that $b_1(1-\alpha)+b_2\alpha<B-E_H$ is given by $
\alpha'=\alpha, \;\;\; \text{w.p.}\;\;1,$ 
%\end{align}
and, 
\begin{align}
&(b_1',b_2')= \left\{ \,
\begin{IEEEeqnarraybox}[][c]{l?s}
\IEEEstrut
(b_1-P\alpha,b_2-P(1-\alpha))& w.p.~$1-p$, \\
(b_1-P\alpha+E_H(1-\alpha),b_2-P(1-\alpha)+E_H\alpha)&w.p.~$p$, \IEEEstrut
\end{IEEEeqnarraybox}
\right.
\label{eq:tpm1}
\end{align}
where we have accounted for the fact that, in a slot, the event that the transmitter harvests and stores $E_H$ units of energy in the charging battery occurs with probability $p$. 
Further, the next state $s'=(b_1',b_2',\alpha')\in \mathcal{S}$ when the system takes action $P=b_1\alpha+b_2(1-\alpha)$ in a state $s=(b_1,b_2,\alpha)\in \mathcal{S}$ such that $B>b_1(1-\alpha)+b_2\alpha\geq B-E_H$  is given by
\begin{align}
(b_1',b_2',\alpha')= \left\{ \,
\begin{IEEEeqnarraybox}[][c]{l?s}
\IEEEstrut
((1-\alpha)b_1,\alpha b_2,\alpha)& w.p.~$1-p$, \\
((1-\alpha)B,\alpha B,1-\alpha)&w.p.~$p$, \IEEEstrut
\end{IEEEeqnarraybox}
\right.
\label{eq:tpm2}
\end{align}
where we have accounted for the constraint that the role of the batteries must be switched when the working battery becomes empty and the charging battery becomes full. Finally, when the system takes action $P=b_1\alpha+b_2(1-\alpha)$ in state $s=(b_1,b_2,\alpha)\in \mathcal{S}$ such that $b_1(1-\alpha)+b_2\alpha=B$, it transitions to the next state $s'=((1-\alpha)B,\alpha B,1-\alpha)$ with probability $1$. 
From \eqref{eq:tpm1} and \eqref{eq:tpm2}, we can easily construct the probability transition matrix, $q(s'|P,s)$ for all $s\in \mathcal{S}$ and $P\in\mathcal{A}(s)$. 

\subsubsection{Optimal value function}
We consider  $K$ slots for the optimization.  
We obtain the optimal value function $V_k(s)$ in slot $k\in\{1,\ldots,K\}$ and state $s\in \mathcal{S}$ by solving the following Bellman equation. \begin{align}\label{eq:bellman}
V_k(s)=\max_{P\in\mathcal{A}(s)}\{\frac{1}{2}\log(1+P)+\sum_{s'\in\mathcal{S}}q(s'|P,s)V_{k+1}(s')\}, &&
\end{align}
for $k=K,K-1,\ldots,1$, where $V_{K+1}(s)\triangleq 0$. 
The optimal online throughput is then given by
\begin{align}\label{eq:on-opt}
T_{\rm on}=\lim_{K\rightarrow \infty}\frac{V_K(s)}{K}.
\end{align}
We note that the optimal online throughput, $T_{\rm on}$, can also be obtained by solving the Bellman equation in the infinite horizon case with the discount factor arbitrarily close to one. Since the reward is bounded and the state space is finite, there exists an optimal stationary  deterministic policy for \eqref{eq:bellman}, i.e., there exists a unique optimal action $P^*(s)$ in state $s$ independent of slot indices \cite{Anup_MDP}. Hence, it suffices to search only in the set of all stationary deterministic policies. 

\subsection{Non-Adaptive (NA) Online Policies}
We note that the above optimal online policy, obtained via dynamic programming, adapts the transmit power in every slot based on the state of the batteries. 
In the non-adaptive policies, we assume the states of batteries are not estimated in every slot. However, we assume that a \emph{flag} is raised when the charging battery becomes full or working battery becomes empty. 
Further,  the energy remaining in the working battery is discarded once the charging battery accumulates $B$ units of energy. 
On the other hand, if the working battery gets completely discharged before the charging battery is full, the transmitter waits without transmission till the charging battery gets full. This implies that the renewal instant is the same as the instant when the charging battery accumulates $B$ units of energy.
Hence, the renewal length, $L$ is distributed identically as $C$ in \eqref{eq:neg_bino}. With this setting, we obtain optimal and suboptimal power allocations in the following. 
%Since in these policies, we do not adapt power allocations based on the battery state, we refer to them as the \emph{non-adaptive} policies. 

\subsubsection{Optimal Non-Adaptive (ONA) Online Policy}
%In this case, our analysis proceeds along the lines in \cite{Shaviv-Ozgur,Rahul_online_TIT}. 
Let $\tilde{P}_i$ be the transmit power in slot $i\in \{1,\ldots,L\}$ after  a renewal, i.e., $\tilde{P}_{i}\triangleq P_{S_k+i-1}$ for $k=1,2,\ldots$. Then, the throughput in \eqref{eq:rate} can be re-written as
\begin{align}
T_{\rm ONA}&=\frac{p}{r}\sum_{n=1}^{\infty}q_n\sum_{i=1}^{n}\frac{1}{2}\log(1+\tilde{P}_i)
=\sum_{i=1}^{\infty}\frac{p}{2r}\bar{F}_{i-1}(r,p)\log(1+\tilde{P}_i),\label{eq:online}
\end{align}
where we recall $\bar{F}_{i-1}(r,p)= \sum_{n=i}^{\infty}q_n$. 
From \eqref{eq:p2p-1}, \eqref{eq:rate} and \eqref{eq:online},  to maximize the long-term average throughput in the online case, we need to solve the following optimization problem.
\begin{subequations}\label{eq:P1}
	\begin{align}
	\underset{\substack{\tilde{P}_i, i=1,2,\ldots}}{\text{maximize}} &\;\;\sum_{i=1}^{\infty}\frac{p}{2r}\bar{F}_{i-1}(r,p)\log(1+\tilde{P}_i)\;\;\;&&\\
	\text{subject to} &\;\; \sum_{i=1}^{\infty}\tilde{P}_i\leq B,\;\;\tilde{P}_i\geq 0,\;\;i=1,2,\ldots\label{eq:P2_c2} && 
	\end{align}
\end{subequations}
Clearly, \eqref{eq:P1} is a convex optimization problem. Hence, the Karush-Kuhn-Tucker (KKT) conditions are necessary and sufficient for optimality. Based on the KKT conditions, we obtain the optimal power allocations in the following theorem. 
\begin{theorem}\label{thm:opt_sol}
Let
\begin{align}\label{eq:N}
N=\max\left\{n:\frac{\sum_{i=1}^{n}\bar{F}_{i-1}(r,p)}{B+n}\leq  \bar{F}_{n-1}(r,p)\right\}. 
\end{align}	
For $k = 1, 2, \ldots$, the optimal transmit power $P^*_{S_k +i -1}$ is given by $P^*_{S_k +i -1} = \tilde{P}_i^{\rm ONA}$, where
%Then, the optimal transmit power in the $i^{\text{th}}$ slot after a renewal instant is given by, 
\begin{align}\label{eq:opt_power}
\tilde{P}^{\rm ONA}_i(B,r,p)=\left\{ 
\begin{array}{l l}
\frac{(B+N)}{\sum_{j=1}^{N}\bar{F}_{j-1}(r,p)}-1 & \text{for $i=1,\ldots,r$, }\\
\frac{(B+N)\bar{F}_{i-1}(r,p)}{\sum_{j=1}^{N}\bar{F}_{j-1}(r,p)}-1 & \text{for $i=r+1,\ldots,N$},\\
0  &   \text{for $i>N$.}
\end{array} \right.
\end{align}
The corresponding throughput under the optimal non-adaptive policy can be obtained by substituting \eqref{eq:opt_power} in \eqref{eq:online} as
\begin{align}
T_{\rm ONA}=\sum_{i=1}^{N}\frac{p}{2r}\bar{F}_{i-1}(r,p)\log\left(\frac{(B+N)\bar{F}_{i-1}(r,p)}{\sum_{j=1}^{N}\bar{F}_{j-1}(r,p)}\right). 
\end{align} 
\end{theorem}
\begin{proof}
See Appendix B. 
\end{proof}
We note  the probability that the charging battery 
fills up in less than $r$ slots after a renewal is zero. Hence, given a fixed amount of energy that can be consumed in the first $r$ slots, it is intuitive that we must consume it at the constant rate in the first $r$ slots after a renewal, as suggested by \eqref{eq:opt_power}. {We also note that under the policy $\tilde{P}_{\rm ONA}$, the working battery becomes fully discharged exactly after $N$ slots. It strikes the optimal balance between discharging too early (i.e., the transmitter will remain idle till the charging battery gets full) and discharging too late (i.e., there is  wastage of energy as the remaining energy in the working battery is discarded). 
Since we are discarding the remaining energy in the working battery in case the charging battery becomes full before $N$ slots and remain idle in case the charging battery takes more than $N$ slots to become full, the optimal solution in Theorem  \ref{thm:opt_sol} is a suboptimal online policy and clearly, we have, $T_{\rm ONA} \le T_{\rm ub}$.} 
%, where $T_{\rm ONA}$ is obtained by substituting $\tilde{P}_i^{\rm ONA}$ from \eqref{eq:opt_power} into \eqref{eq:online}.

\subsubsection{Suboptimal Non-Adaptive (SNA)  Online Policy}
Based on the above ONA policy, we now propose a suboptimal policy, referred to as the  Suboptimal Non-Adaptive (SNA) policy. Our motivations for proposing the SNA policy are it is simpler than the previous policies, and analytically tractable in the sense that we can use it to lower-bound the optimal long-term average throughput in the dual-battery case.
%While preserving the structure of the policy in Lemma \ref{thm:opt_sol}, we propose another policy where 

In the SNA policy, the power allocation in the $i^{\text{th}}$ slot after a renewal instant is given by
\begin{align}\label{eq:approx_prop}
\tilde{P}_i^{\rm SNA}(B,r,p)&=\frac{Bp}{r}\sum_{n=i}^{\infty}q_n=\mu\bar{F}_{i-1}(r,p).
\end{align}
We discard the energy remaining in the working battery when the charging battery becomes full. Since $\sum_{i=1}^{\infty}\bar{F}_{i-1}(r,p)=\sum_{i=1}^{\infty}\sum_{n=i}^{\infty}q_n=\mathbb{E}[C]=r/p$, we note, $\sum_{i=1}^{\infty}\tilde{P}_i^{\rm SNA}=B$. Hence, the power allocation policy in \eqref{eq:approx_prop} does not violate the energy causality constraint. 
Let $T_{\rm SNA}$ denote the long-term average throughput obtained in this strategy. We now have the following result. 
\begin{theorem}\label{thm:approx}
The long-term average throughputs in the ONA and SNA policies are bounded as, 
\begin{align}\label{eq:bounds}
T_{\rm ub}\geq {T_{\rm ONA} \geq T_{\rm SNA}}\stackrel{\text{}}{\geq} T_{\rm ub}-G(r), 
\end{align}
where %$G(r)=1/(2r\ln(2))$.
$G(r)\triangleq\max_p -\frac{p}{2r}\sum_{i=1}^{\infty}\left(\sum_{n=i}^{\infty}q_n\right)\log\left( \sum_{n=i}^{\infty} q_n \right)$ and $q_n={{n-1}\choose{n-r}}p^{r}(1-p)^{n-r}$.  

\end{theorem}
\begin{proof}
See Appendix C.	
\end{proof}

    \begin{figure}[t]
	\centering
	% This file was created by matlab2tikz.
%
%The latest updates can be retrieved from
%  http://www.mathworks.com/matlabcentral/fileexchange/22022-matlab2tikz-matlab2tikz
%where you can also make suggestions and rate matlab2tikz.
%
\begin{tikzpicture}[scale=0.6]

\begin{axis}[%
width=5.425in,
height=2.968in,
at={(0.91in,0.732in)},
scale only axis,
xmode=log,
xmin=1,
xmax=10000,
xminorticks=true,
xlabel style={font=\color{white!15!black}},
xlabel={$r$},
ymode=log,
ymin=0.001,
ymax=1,
yminorticks=true,
ylabel style={font=\color{white!15!black}},
ylabel={Performance Gap},
axis background/.style={fill=white},
legend style={legend cell align=left, align=left, draw=white!15!black}
]
\addplot [color=black, line width=1.5pt]
  table[row sep=crcr]{%
1	0.721\\
2	0.492041978966458\\
4	0.345467406522295\\
30	0.122119777291118\\
78.0000000000001	0.0749236521705588\\
336	0.0357317803053306\\
2336	0.0134659883942673\\
10000	0.0064950637199501\\
};
\addlegendentry{$G(r)$}

\addplot [color=black, dashed, line width=1.5pt]
table[row sep=crcr]{%
	1	0.721\\
	10000	0.00721\\
};
%\addplot [color=black, line width=1.5pt, draw=none, mark=asterisk, mark options={solid, black}]
%  table[row sep=crcr]{%
%1	0.721\\
%2	0.506302377703548\\
%3	0.411721384783787\\
%4	0.355536889969856\\
%7	0.267260800154607\\
%11	0.212238904489625\\
%18	0.165099701141504\\
%30	0.127234071486551\\
%48	0.100115710466347\\
%78.0000000000001	0.0781568445833869\\
%127	0.0609530394855125\\
%207	0.0475105861592282\\
%336	0.0371109637826391\\
%546.000000000001	0.0289712355352499\\
%886	0.0226331012654463\\
%1438	0.0176798360285223\\
%2336	0.0138042913590733\\
%3793	0.0107808751078691\\
%6158	0.00842017321384233\\
%10000	0.00657559815175611\\
%};
\addlegendentry{$0.72/\sqrt{r}$}

\end{axis}
\end{tikzpicture}%
	\caption{Variation of $G(r)$ with $r$.}
	\label{fig:bound}
\end{figure}
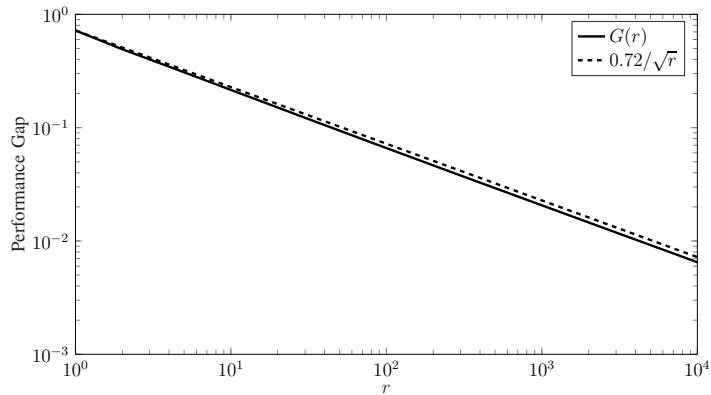

We  make the following remarks on the above theorem. 
\begin{itemize}[leftmargin=*]
	\item Firstly, $G(1)\approx 0.72$ and we recover Proposition 3 of \cite{Shaviv-Ozgur}. 
	\item Further, numerically, we can show that $\max_{r} G(r)=G(1) \approx 0.72$. Hence, the long-term average throughputs in the ONA    and  the SNA policies are at most $0.72$ bits away from the upper bound $T_{\rm ub}$ for any value of $p$, $B$ and $r$, under the Bernoulli energy harvesting model. 
	\item Finally, when the parameters $E_H$ and $p$ are fixed,  $r$ increases  proportionately with $B$,  as $r=B/E_H$. 
	From numerical analysis, we find the bound $G(r)$  
	decreases monotonically at a rate faster than the inverse of the square root of $r$ and $B$, as shown in Fig.    \ref{fig:bound}. 
\end{itemize}

\begin{algorithm}[t]   
	\caption {Power allocation in a renewal period in the proposed adaptive online  policies.}
	\begin{algorithmic}
		\State $i_0=0$.
    	\For {$j=0$ to $r-1$}
    	\For {$k=i_j+1$ to $i_{j+1}$}
    	\State  $\tilde{P}_k^{\rm OA}\leftarrow \tilde{P}_{k-i_j}^{\rm ONA}(B^w_{i_j},r-j,p)$       \Comment{based on optimal non-adaptive policy in \eqref{eq:opt_power}.}	
 	\State  $\tilde{P}_k^{\rm SA}\leftarrow \tilde{P}_{k-i_j}^{\rm SNA}(B^w_{i_j},r-j,p)$       \Comment{based on suboptimal non-adaptive policy in \eqref{eq:approx_prop}.}	
    	\EndFor
    	\EndFor
	\end{algorithmic}
	\label{algo:DB_power_alloc}
\end{algorithm}

\subsection{Adaptive Online Policies}
Here, we assume the state of the charging and working batteries,  denoted by $B^c$ and $B^w$, respectively, are known at the start of every slot. In this policy, we adapt the power allocations based on  $B^c$ and $B^w$. Recall that, by assumption,  it takes $r$ energy arrivals to  fill the charging battery, starting from the empty state. 
Now, until the current slot, suppose that  $j$ energy arrivals have occurred since the last renewal instant, i.e., $B^c=jE_H$. Then, we need $r-j$ more energy arrivals to fill the battery. 
Let $C_{r-j}$ be the random number of slots required for $r-j$ energy arrivals. Then, the complementary CDF of $C_{r-j}$ is given by $\bar{F}_{i}(r-j,p),\; i\in \{1,2,\ldots\}$, where we recall $\bar{F_i}(r,p)=\sum_{n=i+1}^{\infty}q_n$. Now, in the adaptive policy,  we target to allocate $B^w$ units of energy in the working battery based on the distribution of $C_{r-j}$, along the lines in the non-adaptive policies proposed earlier. 
Concretely, the power allocation to $i^{\text{th}}$ slot after $j^{\text{th}}$ energy arrival in a renewal window is given by $\tilde{P}_i^{\rm ONA}(B^w,r-j,p)$, where $\tilde{P}_i^{\rm ONA}(\cdot)$ is defined in \eqref{eq:opt_power}. We refer to this policy as the Optimal-Adaptive (OA) policy. 
Similarly, we can obtain the policy, which we refer to as the  Suboptimal-Adaptive (SA) policy, by allocating $\tilde{P}_i^{\rm SNA}(B^w,r-j,p)$ units of energy to $i^{\text{th}}$ slot after $j^{\text{th}}$ energy arrival in a renewal window,  where $\tilde{P}_i^{\rm SNA}(\cdot)$ is defined in \eqref{eq:approx_prop}. 
We summarize the power allocations in these adaptive policies in Algorithm~\ref{algo:DB_power_alloc}. In the algorithm, $i_j$ represents the index of $j^{\text{th}}$ energy arrival in a renewal window. Further, we represent the long-term average throughputs obtained by OA and SA policies by $T_{\rm OA}$ and  $T_{\rm SA}$, respectively. {
 Since the non-adaptive policies are special cases of adaptive policies, clearly, the optimal performance of the adaptive policies are at least as well as that of the non-adaptive policies. Hence, we have $T_{\rm OA}\geq T_{\rm ONA}$ and $T_{\rm SA}\geq T_{\rm SNA}$.}

\subsection{Constant-Power (CP) Policy}
In the CP policy,  the transmit power remains constant whenever  transmission occurs. Our motivations for proposing this policy are the following. First, practical implementation of such a policy is simple as it does not require the knowledge of battery states. 
Further, prior knowledge of the optimal transmit power can enable system designers to choose appropriate system components such as power amplifiers such that optimal transmit power is in their linear operating region.  
Finally, several variants of this policy, considered in  \cite{Shaviv-Ozgur} and references therein, are shown to perform competitively with optimal policies. Specifically, the  version proposed in \cite{Koksal} has been shown to approach $T_{\rm ub}$ asymptotically with $B$. 
Hence, this policy is also useful in benchmarking the policies studied above.

In the CP policy, we consume the energy available in the battery at a constant rate of $B/(\lfloor rp^{-1} \rfloor)$ as long as the battery is not empty, i.e., the power allocation in the $i^{\text{th}}$ slot after a renewal instant is given by
\begin{align}\label{eq:constant}
\tilde{P}_i&=\frac{B}{\lfloor rp^{-1}\rfloor}, \quad \text{for}\;\; i=1,\ldots, \lfloor r p^{-1}\rfloor.
\end{align}
We discard the remaining energy in the working battery when the charging battery becomes full. 
We denote the long-term average throughput of the policy by $T_{\rm const}$.  %Several variants of this policy have been considered in the literature \cite{Shaviv-Ozgur} and references therein, and they  have been shown to perform competitively with optimal policies. Specifically, the  version proposed in \cite{Koksal} has been shown to approach $T_{\rm ub}$ asymptotically with $B$. 
%{Hence, we compare the above online policies with the CP policy.}

\section{MAC: Single-Battery and Dual-Battery Cases}\label{sec:MAC}
We now consider a $U$-user  MAC where the users communicate to a common receiver over an AWGN channel with  unit noise power.  Let $\mathcal{U}\triangleq \{1,\ldots,U\}$ represent the set of user indices. In the single-battery and dual-battery cases, we assume the user $u\in \mathcal{U}$ is equipped with a single-battery of capacity $2B_u$ units and  with two identical batteries, each of capacity $B_u$ units, respectively.  
Each user applies the half-duplex constraint and the cycle constraint described in Section \ref{sec:system-model}. 
Let $E_u$ be the energy arrival process in user $u$. We assume the energy arrivals are i.i.d. over time and that the amount of energy harvested in slot $i$ in user $u$ is, 
\begin{align}\label{eq:EA-MAC}
E_{ui}&=\left\{ 
\begin{array}{l l}
E_{H_u}  & \text{w.p. $p_u$,}\\
0  &   \text{w.p. $1-p_u$.}\\
\end{array} \right.
\end{align} 
We assume that $B_u/E_{H_u}=r_u$ for some $r_u\in \{1,2,\ldots\}$ and we 
note the average EH rate, $\mu_u =  p_uE_{H_u}$ at  user $u\in\mathcal{U}$. 
The battery evolution at each user is similar to that in the P2P channel case,  described in Section \ref{sec:bat_ev}.  
Let the transmit power in user $u$ in slot $i$ be denoted by $P_{ui}$.
Then, from \cite{OnlineOzgur}, the maximum average throughput region, averaged over all the sample paths of energy arrivals is given by
\begin{align}\label{eq:genCapRegion}
&\mathcal{T}_K\left(P_{uk}\right)
\left\{R_u:\sum_{u\in \mathcal{S}}R_u\leq \frac{1}{K}\mathbb{E}  \left[\sum_{k=1}^{K}\frac{1}{2}\log\left(1+
{\sum_{u\in \mathcal{S}}P_{uk}}\right)\right],\; \forall \mathcal{S}\subseteq \mathcal{U} \right\}.&&
\end{align}
Our goal is to maximize the long-term average throughput region defined as
$\mathcal{T}=\lim_{K\rightarrow \infty}\mathcal{T}_K$.
%\end{align}
We now present an outer bound to $\mathcal{T}$ in the following. 
When the cycle and half-duplex constraints are not present, and the capacities of the batteries are infinite, the largest throughput region for the Gaussian MAC is given by \cite{OnlineOzgur}, 
\begin{align}\label{eq:outer}
&\mathcal{T}_{\rm outer}=\left\{R_u:\sum_{u\in \mathcal{S}}R_u\leq \frac{1}{2}\log\left(1+
{\sum_{u\in \mathcal{S}}\mu_u}\right),\; \forall \mathcal{S}\subseteq \mathcal{U} \right\},&&
\end{align}
where $\mu_u$ is the mean energy arrival rate in user $u\in \mathcal{U}$. 

In the sequel, we propose achievable strategies in the single-battery and the dual-battery cases, based on the P2P channel studied in previous sections. 

\subsection{Single-Battery Case}
In the single-battery case, for simplicity, we only consider the relaxed problem, {where we allow the number of slots over which transmission occurs to take a positive real value.}  
As in the P2P channel under the single-battery case studied in Section \ref{sec:SU-SB}, we assume that transmission occurs with a constant power at  each user. Let $(R_u,P_u)$  be transmit rate and power pairs at user $u\in \mathcal{U}$. 
From \eqref{eq:SB_power}, the long-term average throughput in user $u$ is given by
\begin{align}
T_u=\frac{\mu_uR_u}{\mu_u+P_u}. 
\end{align}
We now note that each boundary point on the largest achievable throughput region is the optimal solution to $\max_{\{T_1,\ldots,T_U\}\in \mathcal{T}}\sum_{u=1}^{U}\lambda_uT_u$ for some $(\lambda_1,\ldots,\lambda_U)$, where $\lambda_u\geq 0,\; u\in \mathcal{U}$ and $\sum_{u=1}^{U}\lambda_u=1$.   
Hence, we  obtain all the boundary points on the largest achievable throughput region in this policy by solving the 
following problem for  different instances of $(\lambda_1,\ldots,\lambda_U)$.
\begin{subequations}\label{eq:SB_2U_AC}
	\begin{align}
	\underset{R_u,P_u}{\text{maximize}} &\;\;\sum_{u=1}^{U}\frac{\lambda_u\mu_uR_{u}}{\mu_u+P_u},\;\;\;\label{eq:obj_SB}&&\\
	\text{subject to} 
	&\;\; \sum_{u\in \mathcal{S}}R_u\leq \frac{1}{2}\log\left(1+\sum_{u\in \mathcal{S}} P_u\right),\; \forall \; \mathcal{S}\subseteq \mathcal{U},\; R_u,P_u\geq 0, \; \forall\; u\in\mathcal{U}. 
	\end{align}
\end{subequations}
Note that the above optimization problem is non-convex as \eqref{eq:obj_SB} is a sum of ratios. In the sequel, we transform the above problem into an equivalent convex optimization problem by applying the \emph{quadratic transform},  recently developed in \cite{quadratic_transform}, briefly described below. Consider a concave $A(x)$ and a convex $B(x)$. Then, \cite{quadratic_transform} shows that  following optimization problems are equivalent.  
\begin{align}\label{eq:QT}
\underset{x\in \mathcal{X}}{\text{maximize}}\; \frac{A(x)}{B(x)} \;\; \stackrel{\text{quadratic transform}}{\Longleftrightarrow}    \;\;\underset{x\in \mathcal{X},y\in \mathbb{R}}{\text{maximize}} \;2y\sqrt{A(x)}-y^2B(x).
\end{align}
We note that the optimization problem in the left-hand side is non-convex and that in the right-hand side is convex. 
Now, applying the quadratic transform to  \eqref{eq:SB_2U_AC}, we obtain the following equivalent convex optimization problem.
\begin{subequations}\label{eq:SB_2U_AC-equi}
	\begin{align}
	\underset{R_u,P_u,y_u}{\text{maximize}} &\;\;\sum_{u=1}^{U}\left(2y_u\sqrt{\lambda_u\mu_uR_{u}}-y_u^2(\mu_u+P_u)\right),&&\\
	\text{subject to} 
	&\;\; \sum_{u\in \mathcal{S}}R_u\leq \frac{1}{2}\log\left(1+\sum_{u\in \mathcal{S}} P_u\right),\; \forall \; \mathcal{S}\subseteq \mathcal{U}, \\ 
	&\;\; R_u,P_u\geq 0,\; y_u\in \mathbb{R}, \; \forall u\in\mathcal{U}. 
	\end{align}
\end{subequations}
The above problem can be solved by an alternate maximization over $\{y_u,u\in \mathcal{U}\}$ and $\{(R_u,P_u),u\in \mathcal{U}\}$, as shown in Algorithm \ref{algo:SB_region}. 
\begin{algorithm}[t]   
	\caption {An iterative algorithm to solve \eqref{eq:SB_2U_AC-equi}.}
	\begin{algorithmic}
        \State Initialize  $\{(R_u,P_u),u\in \mathcal{U}\}$ to a feasible value.
        \State \textbf {Step 1:} Update $y_u^*\leftarrow {\sqrt{\lambda_u\mu_uR_{u}}}/{(\mu_u+P_u)}$.
		\State \textbf {Step 2:} Update $\{(R_u,P_u),u\in \mathcal{U}\}$ by solving  \eqref{eq:SB_2U_AC-equi} with $y_u=y_u^*$.
		\State  Repeat Step 1 and Step 2 until convergence. 
 	\end{algorithmic}
 \label{algo:SB_region}
\end{algorithm}
 Step 1 in Algorithm \ref{algo:SB_region} is because, for a fixed $\{(R_u,P_u),u\in \mathcal{U}\}$, the optimal $y_u^*$ in  \eqref{eq:SB_2U_AC-equi} is given by ${\sqrt{\lambda_u\mu_uR_{u}}}/{(\mu_u+P_u)}$. 
Further, the convergence of the above alternate maximization to the global optimum follows from Theorem 3 in \cite{quadratic_transform}. 
%We note that the largest achievable rate region obtained from Algorithm \ref{algo:SB_region} 

\subsection{Dual-Battery Case}
%We now consider a $U$-user  MAC where the users communicate to a common receiver over an AWGN channel with  unit noise power.  Let $\mathcal{U}\triangleq \{1,\ldots,U\}$ represent the set of user indices. The user $u\in \mathcal{U}$ is equipped with two identical batteries of capacity $B_u$.  
We now consider the dual-battery case and present the following result on the inner region. To obtain the inner region, we assume each user adopts the SNA  policy in \eqref{eq:approx_prop} individually. 
\begin{proposition}\label{thm:inner-MAC}
The long-term average throughput region, $\mathcal{T}$ is bounded by $\mathcal{T}_{\rm inner\_DB}\subseteq\mathcal{T}$, where, 
\begin{flalign}\label{eq:inner}
&\mathcal{T}_{\rm inner\_DB}=\left\{ 
\begin{array}{l}
R_u:\sum_{u\in \mathcal{S}}R_u\leq \frac{1}{2}\log\left(1+
{\sum_{u\in \mathcal{S}}\mu_u}-\min(G(r_u))\right), \forall \mathcal{S}\subseteq \mathcal{U}   
\end{array} \right\}.
\end{flalign} 
\end{proposition}
\begin{proof}
The result can be obtained by extending Theorem \ref{thm:approx} along the lines in the proof of Theorem 1 in \cite{OnlineOzgur}. 
%See Appendix D. 
\end{proof}
%As in the single P2P case, from numerical analysis, we find the bound $G(r_u)$  
%decreases monotonically at a rate faster than the inverse of the square root of $r_u$ and $B_u$. 

We now propose an achievable scheme based on the proposed ONA policy in the P2P channel case.  Let $R_{ui}$  and $P_{ui}$ be the transmit rate and power in slot $i$ of user $u$. Then, from \eqref{eq:online}, we recall that the long-term average throughput in user $u$ is given by  $\sum_{i=1}^{\infty}({p_u}/{r_u})\bar{F}_{u(i-1)}R_{ui}$, where  $\bar{F}_{u(i-1)}=1-F_{ui}=\sum_{n=i+1}^{\infty}q_{un}$  and   $q_{un}\triangleq {{n-1}\choose{n-r_u}}p_u^{r_u}(1-p_u)^{n-r_u}$.  
Hence, along the lines in the previous subsection, we  obtain all the boundary points on the largest achievable throughput region by solving the 
following convex optimization problem for  different instances of $(\lambda_1,\ldots,\lambda_U)$ with $\sum_{u=1}^{U}\lambda_u=1$.
\begin{subequations}\label{eq:2B-ARC}
	\begin{align}
\underset{\substack{R_{ui}, P_{ui}}}{\text{maximize}} &\;\;\sum_{i=1}^{\infty}\sum_{u=1}^{U}\frac{\lambda_u p_u}{r_u}\bar{F}_{u(i-1)}R_{ui},\;\;\;&&\\
	\text{subject to} &\;\; \sum_{i=1}^{\infty}P_{ui}\leq B_u,\; R_{ui}, P_{ui}\geq 0&&\\
	&\sum_{u\in \mathcal{S}}R_{ui}\leq \frac{1}{2}\log\left(1+\sum_{u\in \mathcal{S}}P_{ui}\right),\; \forall \; \mathcal{S}\subseteq \mathcal{U},&&
%	& R_{2j}\leq \frac{1}{2}\log(1+P_{2j}),&&\\
%	&R_{1i}+R_{2j}\leq \frac{1}{2}\log(1+P_{1i}+P_{2j}),\label{eq:all_power} &&\\
%	&R_{ui}, P_{ui}\geq 0,  && 
	\end{align}
\end{subequations}
for all $u\in \mathcal{U}$ and  $i\in \{1,2,\ldots\}$. 
The above problem is a convex, and hence, it can be solved efficiently using standard numerical techniques. 
In this policy,  we note that  user $u$ transmits with power $P_{ui}$ and rate $R_{ui}$ in slot $i$ after a renewal, independent of rates and transmit powers in other users. 

\begin{figure}[t]
	\centering
	% This file was created by matlab2tikz.
%
%The latest updates can be retrieved from
%  http://www.mathworks.com/matlabcentral/fileexchange/22022-matlab2tikz-matlab2tikz
%where you can also make suggestions and rate matlab2tikz.
%

\definecolor{mycolor1}{rgb}{0.00000,0.44700,0.74100}%
\definecolor{mycolor2}{rgb}{0.85000,0.32500,0.09800}%
\definecolor{mycolor3}{rgb}{0.92900,0.69400,0.12500}%
\definecolor{mycolor4}{rgb}{0.49400,0.18400,0.55600}%
\definecolor{mycolor5}{rgb}{0.46600,0.67400,0.18800}%
\begin{tikzpicture}[scale=0.6]

\begin{axis}[%
width=5.425in,
height=3.135in,
at={(0.91in,0.565in)},
scale only axis,
xmin=1,
xmax=22,
xlabel style={font=\color{white!15!black}},
xlabel={$B$ \si{\joule}},
ymin=0.15,
ymax=0.3,
ylabel style={font=\color{white!15!black}},
ylabel={Average Throughput, $T$ (bps)},
axis background/.style={fill=white},
xmajorgrids,
ymajorgrids,
legend style={at={(0.97,0.03)}, anchor=south east, legend cell align=left, align=left, draw=white!15!black}
]
\addplot [color=mycolor1, line width=2.0pt, mark=diamond, mark options={ mycolor1},mark size=4pt]
  table[row sep=crcr]{%
	1	0.292481250360577\\
	2	0.292481250360577\\
	3	0.292481250360577\\
	5	0.292481250360577\\
	8	0.292481250360577\\
	13	0.292481250360577\\
	22	0.292481250360577\\
};
\addlegendentry{$T_{\rm ub}$}

\addplot [color=mycolor2, line width=2.0pt, mark=square, mark options={ mycolor2}]
  table[row sep=crcr]{%
	1	0.277940896830952\\
	2	0.283909909034154\\
	3	0.286438861073343\\
	5	0.28869631546868\\
	8	0.290061019543014\\
	13	0.290971110837884\\
	22	0.291581146561985\\
};
\addlegendentry{$T_{\rm off}$}

\addplot [color=mycolor3, line width=2.0pt, mark=triangle, mark options={ mycolor3},mark size=4pt]
  table[row sep=crcr]{%
	1	0.249999979357096\\
	2	0.265303322717646\\
	3	0.272551598128871\\
	5	0.278981856668377\\
	8	0.28303738167293\\
	13	0.285888274620984\\
	22	0.287923835096297\\
};
\addlegendentry{$T_{\rm on}$}

\addplot [color=mycolor4, line width=2.0pt, mark=asterisk, mark options={ mycolor4},mark size=4pt]
  table[row sep=crcr]{%
	1	0.25\\
	2	0.257573398153568\\
	3	0.263018353726956\\
	5	0.269328304742736\\
	8	0.274171327297292\\
	13	0.278164394171434\\
	22	0.281533111216362\\
};
\addlegendentry{$T_{\rm ONA}$}

\addplot [color=mycolor5, line width=2.0pt, dotted]
  table[row sep=crcr]{%
	1	0.200762077884654\\
	2	0.222201181263419\\
	3	0.233691557919808\\
	5	0.24614192616265\\
	8	0.255523032843882\\
	13	0.263343269298506\\
	22	0.270020633462195\\
};
\addlegendentry{$T_{\rm  SNA }$}

\addplot [color=black, line width=2.0pt, mark=otimes, mark options={black},mark size=4pt]
  table[row sep=crcr]{%
	1	0.166666666666668\\
	2	0.166689875636788\\
	3	0.167279929963225\\
	5	0.167276251862265\\
	8	0.167320645975138\\
	13	0.167309342835591\\
	22	0.167324328034251\\
};
\addlegendentry{$T_{\rm SB}$}

\end{axis}
\end{tikzpicture}%
	\caption{Variation of the optimal long-term average throughput with $B$ for $p=0.1$, $E_H=1$ and with $r=B$.}
	\label{fig:TvsB}
\end{figure}
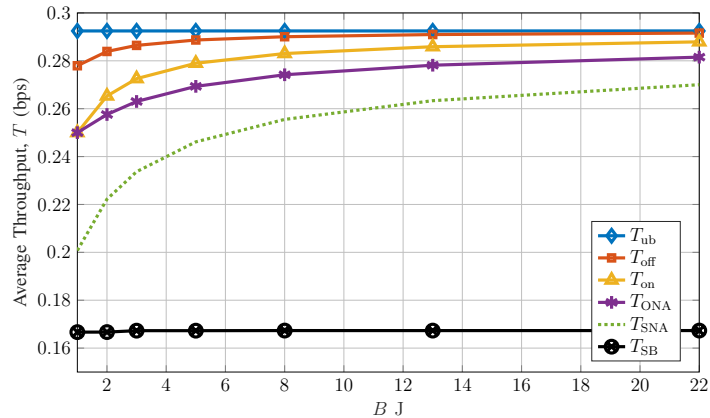
\section{Numerical Results}\label{sec:numerical-results}
%We assume the slot duration is $1$ second and that $1$ \si{\joule} energy arrives with probability $p$ in every slot, i.e., $B=r$. 
In this section, we first compare the long-term average throughputs obtained in the P2P channel under single-battery and dual-battery cases. We then obtain long-term average throughput regions in a MAC using the schemes presented in the previous section. {The parameters used for our simulations are in the similar range as in \cite{Shaviv-Ozgur,MAC-Pillai-letter,MAC-ulukus}}. 

\subsection{Long-Term Average Throughput in a P2P Channel}
In Fig. \ref{fig:TvsB}, we plot variation of the long-term average throughput with the battery capacity in various policies. We note that as the battery capacity increases, the performance gap between offline and online policies decreases. 
Further, the average throughput achieved by the ONA policy approaches the upper bound, $T_{\rm ub}$, much faster than the SNA policy, as the battery capacity $B$ increases. 
%Further, the performance gap of the ONA policy from the upper bound, $T_{\rm ub}$ tends to zero as $B$ increases, much faster than the SNA policy. 
We also note the long-term average throughput in the single battery case,  $T_{\rm SB}$ does not depend on the battery capacity. This is because, as seen in \eqref{eq:SB-real}, the maximum long-term average throughput for the relaxed problem in \eqref{eq:SB_power} depends only on the average harvested power; rounding the solution introduces a negligible change in the throughput of the relaxed problem.  

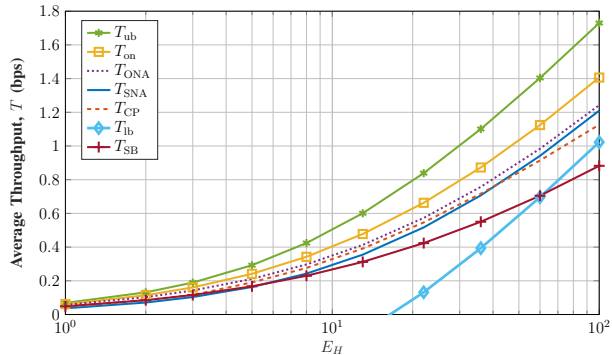
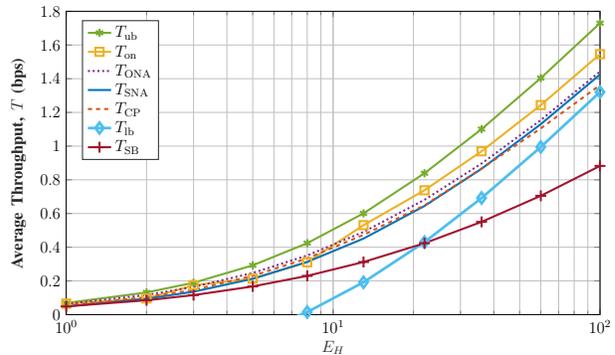
\begin{figure}[t]
	\centering
	\begin{subfigure}{0.48\textwidth}
		% This file was created by matlab2tikz.
%
%The latest updates can be retrieved from
%  http://www.mathworks.com/matlabcentral/fileexchange/22022-matlab2tikz-matlab2tikz
%where you can also make suggestions and rate matlab2tikz.
%
\definecolor{mycolor1}{rgb}{0.00000,0.44700,0.74100}%
\definecolor{mycolor2}{rgb}{0.85000,0.32500,0.09800}%
\definecolor{mycolor3}{rgb}{0.92900,0.69400,0.12500}%
\definecolor{mycolor4}{rgb}{0.49400,0.18400,0.55600}%
\definecolor{mycolor5}{rgb}{0.46600,0.67400,0.18800}%
\definecolor{mycolor6}{rgb}{0.30100,0.74500,0.93300}%
\definecolor{mycolor7}{rgb}{0.63500,0.07800,0.18400}%
\begin{tikzpicture}[scale=0.515]

\begin{axis}[%
width=5.425in,
height=3.085in,
at={(0.91in,0.615in)},
scale only axis,
xmode=log,
xmin=1,
xmax=100,
xminorticks=true,
xlabel style={font=\bfseries\color{white!15!black}},
xlabel={$E_H$},
ymin=0,
ymax=1.8,
ylabel style={font=\bfseries\color{white!15!black}},
ylabel={Average Throughput, $T$ (bps)},
axis background/.style={fill=white},
xmajorgrids,
xminorgrids,
ymajorgrids,
legend style={at={(0.03,0.97)}, anchor=north west, legend cell align=left, align=left, draw=white!15!black}
]
\addplot [color=mycolor5, line width=1.5pt, mark=asterisk, mark options={solid, mycolor5},mark size=3pt]
  table[row sep=crcr]{%
1	0.0687517618749673\\
2	0.131517202916897\\
3	0.189255811626865\\
5	0.292481250360578\\
8	0.423998453277475\\
13	0.600816930584825\\
22	0.839035952556319\\
36	1.10081693058483\\
60	1.4036774610288\\
100	1.72971580931865\\
};
\addlegendentry{$T_{\rm ub}$}

%\addplot [color=mycolor2, line width=1.5pt, mark=o, mark options={solid, mycolor2},mark size=4pt]
%  table[row sep=crcr]{%
%1	0.0651412461930647\\
%2	0.120918665955823\\
%3	0.170613442794331\\
%5	0.257420128720326\\
%8	0.366455227830154\\
%13	0.513081939694767\\
%22	0.713767895471041\\
%36	0.94043298646532\\
%60	1.21117888289516\\
%100	1.51182313958509\\
%};
%\addlegendentry{$T_{\rm off}$}
%
\addplot [color=mycolor3, line width=1.5pt, mark=square, mark options={solid, mycolor3},mark size=3pt]
  table[row sep=crcr]{%
1	0.0619434276645547\\
2	0.114316259114809\\
3	0.160663234115088\\
5	0.241352199518592\\
8	0.342370132351569\\
13	0.477726279139196\\
22	0.662928120453039\\
36	0.872511029683186\\
60	1.12418856624401\\
100	1.40631256154074\\
};
\addlegendentry{$T_{\rm on}$}

\addplot [dotted,color=mycolor4, line width=1.5pt]
  table[row sep=crcr]{%
1	0.0570685794006538\\
2	0.102650870276474\\
3	0.142468058899552\\
5	0.211221220973765\\
8	0.297087444938823\\
13	0.412830162438682\\
22	0.573046271131767\\
36	0.757584065067167\\
60	0.983762644437749\\
100	1.24250734965098\\
};
\addlegendentry{$T_{\rm ONA}$}

\addplot [color=mycolor1, line width=1.5pt]
  table[row sep=crcr]{%
1	0.0367005606062105\\
2	0.071105184411183\\
3	0.103511179706159\\
5	0.163246117667119\\
8	0.242665687783123\\
13	0.355180529105423\\
22	0.516800343305792\\
36	0.706890282002763\\
60	0.941616643427597\\
100	1.20981548289675\\
};
\addlegendentry{$T_{\rm SNA}$}

\addplot [color=mycolor2, dashed, line width=1.5pt]
table[row sep=crcr]{%
	1	0.0447795047902773\\
	2	0.0856599897575183\\
	3	0.115934232540293\\
	5	0.190499344226354\\
	8	0.276159333983872\\
	13	0.391325020442838\\
	22	0.546482205431164\\
	36	0.716985800392838\\
	60	0.914245393513751\\
	100	1.12660119910911\\
};
\addlegendentry{$T_{\rm CP}$}

\addplot [color=mycolor6, line width=2pt, mark=diamond, mark options={solid, mycolor6},mark size=4pt]
  table[row sep=crcr]{%
13	-0.106289850601722\\
22	0.131929171369772\\
36	0.393710149398278\\
60	0.696570679842255\\
100	1.0226090281321\\
};
\addlegendentry{$T_{\rm lb}$}

\addplot [color=mycolor7, line width=1.5pt, mark=+, mark options={solid, mycolor7},mark size=4pt]
  table[row sep=crcr]{%
1	0.0487468750600963\\
2	0.0850344916345622\\
3	0.115770251344156\\
5	0.167276251862266\\
8	0.229838947077224\\
13	0.312189583475379\\
22	0.423753465690782\\
36	0.550480453953168\\
60	0.705059485541764\\
100	0.882447735197915\\
};
\addlegendentry{$T_{\rm SB}$}

\end{axis}
\end{tikzpicture}%
			\caption{For $B=E_H$, i.e, when a single energy arrival is required to completely fill the battery.}
		\label{fig:BeqE}
	\end{subfigure} 
	\hfill
	\begin{subfigure}{0.48\textwidth}
		% This file was created by matlab2tikz.
%
%The latest updates can be retrieved from
%  http://www.mathworks.com/matlabcentral/fileexchange/22022-matlab2tikz-matlab2tikz
%where you can also make suggestions and rate matlab2tikz.
%
\definecolor{mycolor1}{rgb}{0.00000,0.44700,0.74100}%
\definecolor{mycolor2}{rgb}{0.85000,0.32500,0.09800}%
\definecolor{mycolor3}{rgb}{0.92900,0.69400,0.12500}%
\definecolor{mycolor4}{rgb}{0.49400,0.18400,0.55600}%
\definecolor{mycolor5}{rgb}{0.46600,0.67400,0.18800}%
\definecolor{mycolor6}{rgb}{0.30100,0.74500,0.93300}%
\definecolor{mycolor7}{rgb}{0.63500,0.07800,0.18400}%
\begin{tikzpicture}[scale=0.515]

\begin{axis}[%
width=5.425in,
height=3.085in,
at={(0.91in,0.615in)},
scale only axis,
xmode=log,
xmin=1,
xmax=100,
xminorticks=true,
xlabel style={font=\bfseries\color{white!15!black}},
xlabel={$E_H$},
ymin=0,
ymax=1.8,
ylabel style={font=\bfseries\color{white!15!black}},
ylabel={Average Throughput, $T$ (bps)},
axis background/.style={fill=white},
xmajorgrids,
xminorgrids,
ymajorgrids,
legend style={at={(0.03,0.97)}, anchor=north west, legend cell align=left, align=left, draw=white!15!black}
]
\addplot [color=mycolor5, line width=1.5pt, mark=asterisk, mark options={solid, mycolor5},mark size=3pt]
  table[row sep=crcr]{%
1	0.0687517618749673\\
2	0.131517202916897\\
3	0.189255811626865\\
5	0.292481250360578\\
8	0.423998453277475\\
13	0.600816930584825\\
22	0.839035952556319\\
36	1.10081693058483\\
60	1.4036774610288\\
100	1.72971580931865\\
};
\addlegendentry{$T_{\rm ub}$}

%\addplot [color=mycolor2, line width=1.5pt, mark=o, mark options={solid, mycolor2},mark size=4pt]
%  table[row sep=crcr]{%
%1	0.0676888287153217\\
%2	0.12814160703285\\
%3	0.183030787871027\\
%5	0.280119469754472\\
%8	0.402883270310218\\
%13	0.567700908130614\\
%22	0.790989378574813\\
%36	1.03903434995637\\
%60	1.3296723515967\\
%100	1.64634420389844\\
%};
%\addlegendentry{$T_{\rm off}$}
%
\addplot [color=mycolor3, line width=1.5pt, mark=square, mark options={solid, mycolor3},mark size=3pt]
  table[row sep=crcr]{%
1	0.065366634257372\\
2	0.0958409109550002\\
3	0.174243119322973\\
5	0.213307438633062\\
8	0.310708980607424\\
13	0.530893934004616\\
22	0.737922390231605\\
36	0.969414123285842\\
60	1.24351464049805\\
100	1.54551160336517\\
};
\addlegendentry{$T_{\rm on}$}

\addplot [dotted,color=mycolor4, line width=1.5pt]
  table[row sep=crcr]{%
1	0.0631717978993596\\
2	0.116701032504938\\
3	0.164279391389413\\
5	0.247216550390003\\
8	0.351157273474786\\
13	0.49064500019764\\
22	0.681305144150236\\
36	0.896720353707426\\
60	1.1547008926645\\
100	1.44258304489203\\
};
\addlegendentry{$T_{\rm ONA}$}

\addplot [color=mycolor1, line width=1.5pt]
  table[row sep=crcr]{%
1	0.0487659291698144\\
2	0.0939619675589376\\
3	0.136100049009201\\
5	0.212751169307197\\
8	0.312790395752203\\
13	0.451311221483303\\
22	0.64475537730782\\
36	0.865527829402721\\
60	1.1303015590157\\
100	1.42477416800025\\
};
\addlegendentry{$T_{\rm SNA}$}

\addplot [color=mycolor2,dashed, line width=1.5pt]
table[row sep=crcr]{%
	1	0.0541434467064152\\
	2	0.103572540876228\\
	3	0.146298942585601\\
	5	0.230335086107655\\
	8	0.333907626983883\\
	13	0.473155866472066\\
	22	0.648593412245138\\
	36	0.866916296967512\\
	60	1.10542528266291\\
	100	1.36218728342418\\
};
\addlegendentry{$T_{\rm CP}$}

\addplot [color=mycolor6,  line width=2pt, mark=diamond, mark options={solid, mycolor6},mark size=4pt]
  table[row sep=crcr]{%
	5	-0.116318749639422\\
	8	0.0151984532774749\\
	13	0.192016930584825\\
	22	0.430235952556319\\
	36	0.692016930584825\\
	60	0.994877461028802\\
	100	1.32091580931865\\
};
\addlegendentry{$T_{\rm lb}$}

\addplot [color=mycolor7, line width=1.5pt, mark=+, mark options={solid, mycolor7},mark size=4pt]
  table[row sep=crcr]{%
1	0.0487489132803591\\
2	0.0850571988046296\\
3	0.115770251344156\\
5	0.167324391289705\\
8	0.22990053824128\\
13	0.312189583475379\\
22	0.423753465690782\\
36	0.550521030060641\\
60	0.705091108788902\\
100	0.882447735197915\\
};
\addlegendentry{$T_{\rm SB}$}

\end{axis}
\end{tikzpicture}%
		\caption{For $B=3E_H$, i.e, when two energy arrivals are required to completely fill the battery.}
		\label{fig:Beq2E}
	\end{subfigure}	
	\caption{Variation of the optimal long-term average throughput with $E_H$ when $B=E_Hr$  for $r=1,3$ and $p=0.1$.}
	\label{fig:ONA_SNA}
\end{figure}

{In Fig. \ref{fig:ONA_SNA}, we plot variation of the long-term average throughput with the amount of energy harvested per arrival, $E_H$, for $B=E_H$ (see Fig. \ref{fig:BeqE})  and $B=3E_H$ (see Fig. \ref{fig:Beq2E}). 
When $B=E_H$, a single energy arrival completely fills up the battery. In this scenario,  the system model of the current paper in the dual-battery case is identical to that in \cite{Shaviv-Ozgur}, and Fig. \ref{fig:BeqE} is similar to Fig. 4 of \cite{Shaviv-Ozgur}, where ONA and SNA policies of the current paper correspond to the optimal policy and the Constant-Fraction policy of \cite{Shaviv-Ozgur}, respectively. 
In this case, the long-term average throughputs in ONA and SNA policies are at most $0.72$ bits away from the upper-bound, as pointed out in the remarks on Theorem \ref{thm:opt_sol}.  
In Fig. \ref{fig:Beq2E}, we set $B=3E_H$. For a given $E_H$, note that the mean value of the harvested energy, $\mu=pE_H$, is the same in both the figures. Hence, the upper-bound remains the same in both the figures. However, the performance of the optimal offline, optimal online, ONA and SNA policies when $B=3E_H$ are better than that when $B=E_H$. Based on Theorem \ref{thm:opt_sol}, the long-term average throughputs in ONA and SNA policies are at most $0.41$ bits away from the upper-bound when $B=3E_H$, as illustrated by the $T_{\rm lb}$ curve.    
%This is because, the battery with a higher capacity enables smoother averaging of the harvested energy. 
From Fig. \ref{fig:BeqE}, we also note that the CP policy performs better than the SNA policy when $E_H$ values are small. However, as $E_H$ increases, the performance gap  of the CP policy from the upper-bound increases, unlike the SNA and ONA policies, which maintain a bounded gap from the upper-bound. This observation had been made for $B=E_H$ case in \cite{Shaviv-Ozgur}. A similar observation holds when $B=3E_H$. 
In the single battery case, the performance curves are \emph{almost identical} in both the figures. This shows that, as highlighted Section \ref{sec:SB_opt}, increasing $B$ has a negligible impact on the performance, for any $B\geq E_H$. We also note that the performance gap of the single-battery case from the upper-bound diverges as $E_H$ increases.}

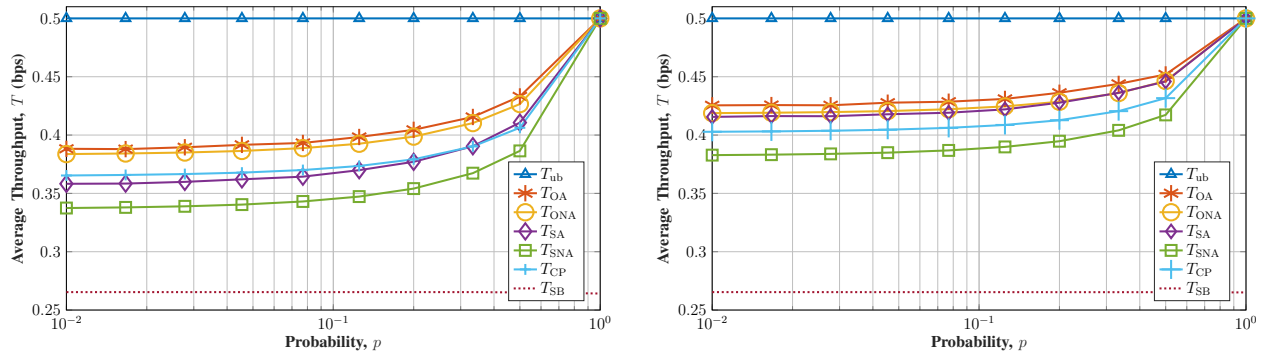
\begin{figure}[t]
	\centering
	\begin{subfigure}{0.48\textwidth}
		% This file was created by matlab2tikz.
%
%The latest updates can be retrieved from
%  http://www.mathworks.com/matlabcentral/fileexchange/22022-matlab2tikz-matlab2tikz
%where you can also make suggestions and rate matlab2tikz.
%
\definecolor{mycolor1}{rgb}{0.00000,0.44700,0.74100}%
\definecolor{mycolor2}{rgb}{0.85000,0.32500,0.09800}%
\definecolor{mycolor3}{rgb}{0.92900,0.69400,0.12500}%
\definecolor{mycolor4}{rgb}{0.49400,0.18400,0.55600}%
\definecolor{mycolor5}{rgb}{0.46600,0.67400,0.18800}%
\definecolor{mycolor6}{rgb}{0.30100,0.74500,0.93300}%
\definecolor{mycolor7}{rgb}{0.63500,0.07800,0.18400}%
\begin{tikzpicture}[scale=0.515]

\begin{axis}[%
width=5.425in,
height=3.085in,
at={(0.91in,0.615in)},
scale only axis,
xmode=log,
xmin=0.01,
xmax=1,
xminorticks=true,
xlabel style={font=\bfseries\color{white!15!black}},
xlabel={Probability, $p$},
ymin=0.25,
ymax=0.51,
ylabel style={font=\bfseries\color{white!15!black}},
ylabel={Average Throughput, $T$ (bps)},
axis background/.style={fill=white},
xmajorgrids,
xminorgrids,
ymajorgrids,
legend style={at={(0.97,0.03)}, anchor=south east, legend cell align=left, align=left, draw=white!15!black}
]
\addplot [color=mycolor1, line width=1.5pt, mark=triangle, mark options={solid, mycolor1,mark size=4pt}]
  table[row sep=crcr]{%
	0.01	0.5\\
	0.0166666666666667	0.5\\
	0.0277777777777778	0.5\\
	0.0454545454545455	0.5\\
	0.0769230769230769	0.5\\
	0.125	0.5\\
	0.2	0.5\\
	0.333333333333333	0.5\\
	0.5	0.5\\
	1	0.5\\
};
\addlegendentry{$T_{\rm ub}$}

\addplot [color=mycolor2, line width=1.5pt, mark=asterisk, mark options={solid, mycolor2,mark size=6pt}]
  table[row sep=crcr]{%
0.01	0.388223334036803\\
0.0166666666666667	0.387945277757083\\
0.0277777777777778	0.389525585132257\\
0.0454545454545455	0.391606971217056\\
0.0769230769230769	0.393228695990357\\
0.125	0.398343831692613\\
0.2	0.404477512929089\\
0.333333333333333	0.41551159205395\\
0.5	0.433106346957335\\
1	0.5\\
};
\addlegendentry{$T_{\rm OA}$}

\addplot [color=mycolor3, line width=1.5pt, mark=o, mark options={solid, mycolor3,mark size=6pt}]
  table[row sep=crcr]{%
0.01	0.38369394464205\\
0.0166666666666667	0.384194096419237\\
0.0277777777777778	0.385031373511197\\
0.0454545454545455	0.386373048889519\\
0.0769230769230769	0.388791494477144\\
0.125	0.392565321364758\\
0.2	0.398664852846782\\
0.333333333333333	0.410263464099105\\
0.5	0.42653659894779\\
1	0.5\\
};
\addlegendentry{$T_{\rm ONA}$}

\addplot [color=mycolor4, line width=1.5pt, mark=diamond, mark options={solid, mycolor4,mark size=6pt}]
  table[row sep=crcr]{%
0.01	0.358224627754016\\
0.0166666666666667	0.358413960990988\\
0.0277777777777778	0.359878965930674\\
0.0454545454545455	0.362044304275492\\
0.0769230769230769	0.364347078745504\\
0.125	0.369909839680236\\
0.2	0.377055730861759\\
0.333333333333333	0.390409929812846\\
0.5	0.410664917739458\\
1	0.5\\
};
\addlegendentry{$T_{\rm SA}$}

\addplot [color=mycolor5, line width=1.5pt, mark=square, mark options={solid, mycolor5,mark size=4pt}]
  table[row sep=crcr]{%
0.01	0.337443787012435\\
0.0166666666666667	0.337992978047736\\
0.0277777777777778	0.338913502941772\\
0.0454545454545455	0.340391636065093\\
0.0769230769230769	0.343065908380691\\
0.125	0.347263475382663\\
0.2	0.354107573942067\\
0.333333333333333	0.367331808292354\\
0.5	0.38643262197434\\
1	0.5\\
};
\addlegendentry{$T_{\rm SNA}$}

\addplot [color=mycolor6, line width=1.5pt, mark=+, mark options={solid, mycolor6,mark size=4pt}]
  table[row sep=crcr]{%
0.01	0.365343660117503\\
0.0166666666666667	0.365798829083161\\
0.0277777777777778	0.366562023767516\\
0.0454545454545455	0.367788212385094\\
0.0769230769230769	0.370008784643647\\
0.125	0.373499549620091\\
0.2	0.3792040448\\
0.333333333333333	0.390260631001372\\
0.5	0.40625\\
1	0.5\\
};
\addlegendentry{$T_{\rm CP}$}

\addplot [color=mycolor7, dotted, line width=1.5pt]
  table[row sep=crcr]{%
0.01	0.265368879883378\\
0.333333333333333	0.265368635255312\\
0.5	0.265098389087256\\
1	0.264160416786859\\
};
\addlegendentry{$T_{\rm SB}$}

\end{axis}
\end{tikzpicture}%
		\caption{For $B=2E_H$, i.e, when two energy arrivals are required to completely fill up the battery.}
		\label{fig:burst_2}
	\end{subfigure} 
	\hfill
	\begin{subfigure}{0.48\textwidth}
		% This file was created by matlab2tikz.
%
%The latest updates can be retrieved from
%  http://www.mathworks.com/matlabcentral/fileexchange/22022-matlab2tikz-matlab2tikz
%where you can also make suggestions and rate matlab2tikz.
%
\definecolor{mycolor1}{rgb}{0.00000,0.44700,0.74100}%
\definecolor{mycolor2}{rgb}{0.85000,0.32500,0.09800}%
\definecolor{mycolor3}{rgb}{0.92900,0.69400,0.12500}%
\definecolor{mycolor4}{rgb}{0.49400,0.18400,0.55600}%
\definecolor{mycolor5}{rgb}{0.46600,0.67400,0.18800}%
\definecolor{mycolor6}{rgb}{0.30100,0.74500,0.93300}%
\definecolor{mycolor7}{rgb}{0.63500,0.07800,0.18400}%
\begin{tikzpicture}[scale=0.515]

\begin{axis}[%
width=5.425in,
height=3.085in,
at={(0.91in,0.615in)},
scale only axis,
xmode=log,
xmin=0.01,
xmax=1,
xminorticks=true,
xlabel style={font=\bfseries\color{white!15!black}},
xlabel={Probability, $p$},
ymin=0.25,
ymax=0.51,
ylabel style={font=\bfseries\color{white!15!black}},
ylabel={Average Throughput, $T$ (bps)},
axis background/.style={fill=white},
xmajorgrids,
xminorgrids,
ymajorgrids,
legend style={at={(0.97,0.03)}, anchor=south east, legend cell align=left, align=left, draw=white!15!black}
]
\addplot [color=mycolor1, line width=1.5pt, mark=triangle, mark options={solid, mycolor1,mark size=4pt}]
  table[row sep=crcr]{%
	0.01	0.5\\
	0.0166666666666667	0.5\\
	0.0277777777777778	0.5\\
	0.0454545454545455	0.5\\
	0.0769230769230769	0.5\\
	0.125	0.5\\
	0.2	0.5\\
	0.333333333333333	0.5\\
	0.5	0.5\\
	1	0.5\\
};
\addlegendentry{$T_{\rm ub}$}

\addplot [color=mycolor2, line width=1.5pt, mark=asterisk, mark options={solid, mycolor2,mark size=6pt}]
  table[row sep=crcr]{%
0.01	0.425462920450805\\
0.0166666666666667	0.425697540683727\\
0.0277777777777778	0.4255242759539\\
0.0454545454545455	0.427726301045507\\
0.0769230769230769	0.428577999739853\\
0.125	0.430985578241915\\
0.2	0.436255739942532\\
0.333333333333333	0.443858030893386\\
0.5	0.451966797585492\\
1	0.5\\
};
\addlegendentry{$T_{\rm OA}$}

\addplot [color=mycolor3, line width=1.5pt, mark=o, mark options={solid, mycolor3,mark size=6pt}]
  table[row sep=crcr]{%
0.01	0.418803711431974\\
0.0166666666666667	0.419130449622136\\
0.0277777777777778	0.419677172946397\\
0.0454545454545455	0.420552637410455\\
0.0769230769230769	0.422128840763327\\
0.125	0.424583345280798\\
0.2	0.428534984998329\\
0.333333333333333	0.435996690226855\\
0.5	0.446348341392132\\
1	0.5\\
};
\addlegendentry{$T_{\rm ONA}$}

\addplot [color=mycolor4, line width=1.5pt, mark=diamond, mark options={solid, mycolor4,mark size=5pt}]
  table[row sep=crcr]{%
0.01	0.415570782639465\\
0.0166666666666667	0.41625696009658\\
0.0277777777777778	0.416151278753571\\
0.0454545454545455	0.417813510091225\\
0.0769230769230769	0.419131220990305\\
0.125	0.422012882958328\\
0.2	0.427754144787857\\
0.333333333333333	0.436207745095518\\
0.5	0.445940686497195\\
1	0.5\\
};
\addlegendentry{$T_{\rm SA}$}

\addplot [color=mycolor5, line width=1.5pt, mark=square, mark options={solid, mycolor5,mark size=4pt}]
  table[row sep=crcr]{%
0.01	0.38279836643393\\
0.0166666666666667	0.383193929769706\\
0.0277777777777778	0.383856505767711\\
0.0454545454545455	0.384919271807731\\
0.0769230769230769	0.386838418346916\\
0.125	0.389841325199978\\
0.2	0.394712783323524\\
0.333333333333333	0.404035247204667\\
0.5	0.417268771633085\\
1	0.5\\
};
\addlegendentry{$T_{\rm SNA}$}

\addplot [color=mycolor6, line width=1.5pt, mark=+, mark options={solid, mycolor6,mark size=7pt}]
  table[row sep=crcr]{%
0.01	0.402806441348669\\
0.0166666666666667	0.403134611401571\\
0.0277777777777778	0.403684453907154\\
0.0454545454545455	0.404566791864072\\
0.0769230769230769	0.406161350092256\\
0.125	0.408659522617513\\
0.2	0.41272023922156\\
0.333333333333333	0.420517987885767\\
0.5	0.431640625\\
1	0.5\\
};
\addlegendentry{$T_{\rm CP}$}

\addplot [color=mycolor7, dotted, line width=1.5pt]
  table[row sep=crcr]{%
0.01	0.265368879883378\\
0.5	0.265307613899834\\
1	0.265098389087256\\
};
\addlegendentry{$T_{\rm SB}$}

\end{axis}
\end{tikzpicture}%
		\caption{For $B=4E_H$, i.e, when four energy arrivals are required to completely fill up the battery.}
		\label{fig:burst_4}
	\end{subfigure}	
	\caption{Variation of the optimal long-term average throughput with $p$ for a fixed mean harvesting rate, $\mu=1$, with $B=rE_H$ for $r=2,4$.}
	\label{fig:burst}
\end{figure}

% Please add the following required packages to your document preamble:
% \usepackage{multirow}

{In Fig. \ref{fig:burst}, we plot variation of the long-term average throughput with $p$, for a fixed mean harvesting rate of $\mu=1$, for $B=2E_H$ (see Fig. \ref{fig:burst_2})  and $B=4E_H$ (see Fig. \ref{fig:burst_4}). For the same parameters, we also present the average idle time, the fraction of the slot length over which the transmit power is zero, in Table \ref{tab:title}, and we plot the variation of the average amount of energy discarded per slot, $E_{\rm discarded}$, with $p$ in Fig. \ref{fig:E_discarded}. 
Note that the lower the value of $p$, the higher is the amount of energy harvested per arrival, i.e., the energy arrival becomes more bursty as $p$ is decreased keeping $\mu$ fixed. We make the following key observations from Fig. \ref{fig:burst}, Table \ref{tab:title} and Fig. \ref{fig:E_discarded}. 
Firstly, the long-term average throughputs in all the policies are significantly lower than the upper-bound when $p$ is small and they approach the upper-bound as $p$ is increased. This is because, as seen from Table \ref{tab:title}, when $p$ is small, the average idle time is large. This indicates that, when $p$ is small, all the policies are \emph{aggressive}, i.e., $B$ units of energy in the working battery is consumed in a smaller duration of time as compared to the case when $p$ is higher. Due to the concavity of the throughput, this leads to a degradation in the throughput. Further, as seen from Fig. \ref{fig:E_discarded}, the average amount of energy discarded when $p$ is small is more than that when $p$ is larger. This also reduces the achievable throughput. 
For the higher values of $p$, the harvested energy arrives more uniformly and the power allocation can be nearly constant, and the average idle time and the average amount of energy discarded decrease. Hence, the performance improves as the $p$ is increased.  
Secondly, when the battery capacity, $B$, is increased from $2E_H$ to $4E_H$, the performance improves. This is because, average idle time and the average amount of energy discarded when $B=4E_H$ are significantly lower than that when $B=2E_H$, as seen from Table \ref{tab:title} and Fig.  \ref{fig:E_discarded}, respectively.}
\begin{table}[t]
	\caption {Average idle time, the fraction of the slot length over which the transmit power is zero (expressed as the percentage of the slot length), for $\mu=1$.} \label{tab:title} 
	\centering
	\begin{tabular}{|l|l|l|l|l|}
		\hline
		Policies & \multicolumn{4}{c|}{Average idle time (\% of the slot length)}                   \\ \hline
		& \multicolumn{2}{c|}{$B=2E_H$} & \multicolumn{2}{c|}{$B=4E_H$} \\ \cline{2-5} 
		& $p=0.01$& $p=0.5$ & $p=0.01$ & $p=0.5$      \\ \hline
		OA  Policy  & $31.0$\%  & $22.0$\%  & $26.0$\%   & $15.0$\%           \\ \hline
		ONA Policy  & $26.0$\%  & $19.0$\%  & $20.0$\%   & $12.0$\%           \\ \hline
		SA  Policy  & $~0.62$\%& $~0.59$\%& $~0.3$\%  & $~0.2$\%           \\ \hline
		SNA   Policy& $~0.2$\% &$~0.15$\% & $~0.17$\% & $~0.05$\%           \\ \hline
		CP  Policy  & $29.0$\%  &$19.0$\%   & $22.0$\%   & $12.0$\%           \\ \hline
		SB Optimal Policy & $63.0$\% &$62.0$\%   & $63.0$\%   & $63.0$\%           \\ \hline
	\end{tabular}
\end{table} 
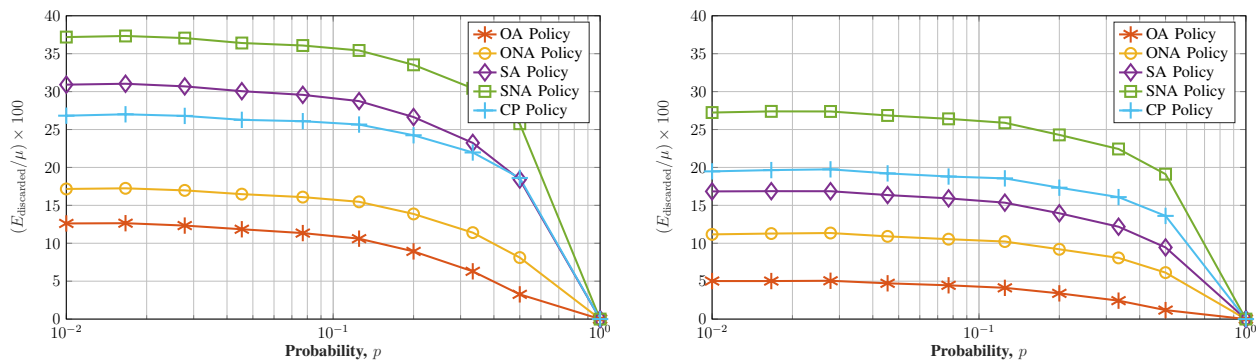
\begin{figure}[!b]
	\centering
	\begin{subfigure}{0.48\textwidth}
		% This file was created by matlab2tikz.
%
%The latest updates can be retrieved from
%  http://www.mathworks.com/matlabcentral/fileexchange/22022-matlab2tikz-matlab2tikz
%where you can also make suggestions and rate matlab2tikz.
%
\definecolor{mycolor1}{rgb}{0.85000,0.32500,0.09800}%
\definecolor{mycolor2}{rgb}{0.92900,0.69400,0.12500}%
\definecolor{mycolor3}{rgb}{0.49400,0.18400,0.55600}%
\definecolor{mycolor4}{rgb}{0.46600,0.67400,0.18800}%
\definecolor{mycolor5}{rgb}{0.30100,0.74500,0.93300}%
\begin{tikzpicture}[scale=0.515]

\begin{axis}[%
width=5.425in,
height=3.085in,
at={(0.91in,0.615in)},
scale only axis,
xmode=log,
xmin=0.01,
xmax=1,
xminorticks=true,
xlabel style={font=\bfseries\color{white!15!black}},
xlabel={Probability, $p$},
ymin=0,
ymax=40,
ylabel style={font=\bfseries\color{white!15!black}},
ylabel={$(E_{\rm discarded}/\mu)\times 100$},
axis background/.style={fill=white},
xmajorgrids,
xminorgrids,
ymajorgrids,
legend style={legend cell align=left, align=left, draw=white!15!black}
]
%legend style={at={(0.97,0.03)}, anchor=south east, legend cell align=left, align=left, draw=white!15!black}
%]
\addplot [color=mycolor1, line width=1.5pt, mark=asterisk, mark options={solid, mycolor1},mark size=6pt]
  table[row sep=crcr]{%
0.01	12.6134300315008\\
0.0166666666666667	12.6448779592051\\
0.0277777777777777	12.3320432455468\\
0.0454545454545454	11.8466581221321\\
0.0769230769230769	11.3405212619981\\
0.125	10.6067474267598\\
0.2	8.92768038774779\\
0.333333333333333	6.29830531457014\\
0.5	3.2617764130938\\
1	0\\
};
\addlegendentry{OA Policy}

\addplot [color=mycolor2, line width=1.5pt, mark=o, mark options={solid, mycolor2},mark size=4pt]
  table[row sep=crcr]{%
0.01	17.1610737366174\\
0.0166666666666667	17.246581533623\\
0.0277777777777777	16.9737510489732\\
0.0454545454545454	16.4774454441133\\
0.0769230769230769	16.087326228158\\
0.125	15.4640137901199\\
0.2	13.8673918138789\\
0.333333333333333	11.4025913894979\\
0.5	8.11457098634672\\
1	0\\
};
\addlegendentry{ONA Policy}

\addplot [color=mycolor3, line width=1.5pt, mark=diamond, mark options={solid, mycolor3},mark size=6pt]
  table[row sep=crcr]{%
0.01	30.9180954863279\\
0.0166666666666667	31.0362359304377\\
0.0277777777777777	30.6888449591049\\
0.0454545454545454	30.0560196311323\\
0.0769230769230769	29.5716253632875\\
0.125	28.7454924082444\\
0.2	26.6451503018372\\
0.333333333333333	23.2346189627129\\
0.5	18.4002373869789\\
1	0\\
};
\addlegendentry{SA Policy}

\addplot [color=mycolor4, line width=1.5pt, mark=square, mark options={solid, mycolor4},mark size=4pt]
  table[row sep=crcr]{%
0.01	37.1908285512867\\
0.0166666666666667	37.3354743124278\\
0.0277777777777777	37.0586659402885\\
0.0454545454545458	36.4085456769223\\
0.0769230769230775	36.0785546876092\\
0.125	35.4256314915697\\
0.2	33.5315116033\\
0.333333333333333	30.4321918688097\\
0.5	25.7647224242153\\
1	0\\
};
\addlegendentry{SNA Policy}

\addplot [color=mycolor5, line width=1.5pt, mark=+, mark options={solid, mycolor5},mark size=6pt]
  table[row sep=crcr]{%
0.01	26.8245527037277\\
0.0166666666666667	27.0075282125507\\
0.0277777777777777	26.7953368768248\\
0.0454545454545454	26.2748259386664\\
0.0769230769230769	26.0983982435306\\
0.125	25.6559452708302\\
0.2	24.2445775694548\\
0.333333333333333	21.9720574599903\\
0.5	18.6116420787145\\
1	0\\
};
\addlegendentry{CP Policy}

\end{axis}
\end{tikzpicture}%
		\caption{For $B=2E_H$, i.e, when two energy arrivals are required to completely fill up the battery.}
		\label{fig:E_discarded_2}
	\end{subfigure} 
	\hfill
	\begin{subfigure}{0.48\textwidth}
		% This file was created by matlab2tikz.
%
%The latest updates can be retrieved from
%  http://www.mathworks.com/matlabcentral/fileexchange/22022-matlab2tikz-matlab2tikz
%where you can also make suggestions and rate matlab2tikz.
%
\definecolor{mycolor1}{rgb}{0.85000,0.32500,0.09800}%
\definecolor{mycolor2}{rgb}{0.92900,0.69400,0.12500}%
\definecolor{mycolor3}{rgb}{0.49400,0.18400,0.55600}%
\definecolor{mycolor4}{rgb}{0.46600,0.67400,0.18800}%
\definecolor{mycolor5}{rgb}{0.30100,0.74500,0.93300}%
\begin{tikzpicture}[scale=0.515]

\begin{axis}[%
width=5.425in,
height=3.085in,
at={(0.91in,0.615in)},
scale only axis,
xmode=log,
xmin=0.01,
xmax=1,
xminorticks=true,
xlabel style={font=\bfseries\color{white!15!black}},
xlabel={Probability, $p$},
ymin=0,
ymax=40,
ylabel style={font=\bfseries\color{white!15!black}},
ylabel={$(E_{\rm discarded}/\mu)\times 100$},
axis background/.style={fill=white},
xmajorgrids,
xminorgrids,
ymajorgrids,
legend style={legend cell align=left, align=left, draw=white!15!black}
]
%legend style={at={(0.97,0.03)}, anchor=south east, legend cell align=left, align=left, draw=white!15!black}
%]
\addplot [color=mycolor1, line width=1.5pt, mark=asterisk, mark options={solid, mycolor1},mark size=6pt]
  table[row sep=crcr]{%
	0.01	5.01619491412857\\
	0.0166666666666667	5.01589074313205\\
	0.0277777777777777	5.05394986041265\\
	0.0454545454545454	4.71817754730741\\
	0.0769230769230769	4.45505098845091\\
	0.125	4.12113175539412\\
	0.2	3.39372012020177\\
	0.333333333333333	2.43150306575118\\
	0.5	1.18408253381008\\
	1	0\\
};
\addlegendentry{OA Policy}

\addplot [color=mycolor2, line width=1.5pt, mark=o, mark options={solid, mycolor2},mark size=4pt]
  table[row sep=crcr]{%
	0.01	11.1719534037229\\
	0.0166666666666667	11.2805803259264\\
	0.0277777777777777	11.3502960167698\\
	0.0454545454545454	10.908744813443\\
	0.0769230769230769	10.5385811301219\\
	0.125	10.2260156045764\\
	0.2	9.20494058588031\\
	0.333333333333333	8.06840052816673\\
	0.5	6.13073817897589\\
	1	0\\
};
\addlegendentry{ONA Policy}

\addplot [color=mycolor3, line width=1.5pt, mark=diamond, mark options={solid, mycolor3},mark size=6pt]
  table[row sep=crcr]{%
	0.01	16.8390050422826\\
	0.0166666666666667	16.869483545652\\
	0.0277777777777777	16.8625119854225\\
	0.0454545454545454	16.3562011945448\\
	0.0769230769230769	15.9182379136164\\
	0.125	15.3536929781386\\
	0.2	13.9555801444029\\
	0.333333333333333	12.1816773261604\\
	0.5	9.45087964712176\\
	1	0\\
};
\addlegendentry{SA Policy}

\addplot [color=mycolor4, line width=1.5pt, mark=square, mark options={solid, mycolor4},mark size=4pt]
  table[row sep=crcr]{%
	0.01	27.2394663561066\\
	0.0166666666666667	27.3918608020033\\
	0.0277777777777777	27.3812097283121\\
	0.0454545454545454	26.8537904665671\\
	0.0769230769230769	26.4094721608668\\
	0.125	25.8869087247301\\
	0.2	24.2915622002178\\
	0.333333333333333	22.4328000496315\\
	0.5	19.1135837795862\\
	1	0\\
};
\addlegendentry{SNA Policy}

\addplot [color=mycolor5, line width=1.5pt, mark=+, mark options={solid, mycolor5},mark size=6pt]
  table[row sep=crcr]{%
	0.01	19.479806159895\\
	0.0166666666666667	19.6491706080016\\
	0.0277777777777777	19.7485274561688\\
	0.0454545454545454	19.215073436885\\
	0.0769230769230769	18.8022926874779\\
	0.125	18.5568003427145\\
	0.2	17.3285173067801\\
	0.333333333333333	16.0854369503422\\
	0.5	13.6093146750299\\
	1	0\\
};
\addlegendentry{CP Policy}

\end{axis}
\end{tikzpicture}%
		\caption{For $B=4E_H$, i.e, when four energy arrivals are required to completely fill up the battery.}
		\label{fig:E_discarded_4}
	\end{subfigure}	
	\caption{Variation of the average amount of energy discarded per slot, $E_{\rm discarded}$ (expressed as the percentage of the mean energy arrival rate, $\mu$),  with $p$ for a fixed mean harvesting rate $\mu=1$, with $B=rE_H$ for $r=2,4$.}
	\label{fig:E_discarded}
\end{figure}
{
	%Further, it is interesting to note from  Table \ref{tab:title} for $p=0.01$, the OA, ONA and CP policies are significantly more aggressive than the SA and SNA policies, as the idle times are almost 30\%. However, in doing so, they minimize the amount of energy wasted, as seen from  Fig.  \ref{fig:E_discarded}. Both the idle time and the amount of energy discarded decrease as $p$ and/or $B$ are increased and hence the long-term average throughput increases, as seen from Fig. \ref{fig:burst}. 
%%
Finally, we note that some policies are more \emph{robust} to burstiness than others.  Specifically, performance of the OA and SA policies are better than that of ONA and SNA policies for all values of $p$. Moreover, the variation of OA and SA policies with $p$ is less than that of the ONA and SA policies. Further, the CP policy is more robust than the SNA policy for the chosen values. 
As expected, the long-term average throughput in the single-battery case does not vary with $p$, as $\mu$ is kept fixed.}

%\begin{figure}[t]
%	\centering
%	\input{TvsE.tex}
%	\caption{Variation of the optimal long-term average throughput with $E$ for $B=Er$, $r=4$ and $p=0.1$.}
%	\label{fig:TvsE}
%\end{figure}

%\begin{figure}[t]
%	\centering
%	\input{burst_long.tex}
%	\caption{Variation of the optimal long-term average throughput with $p$ for a fixed mean harvesting rate $\mu=1$, for $B=1000$.}
%	\label{fig:burst}
%\end{figure}
\subsection{Long-Term Average Throughput Regions in a MAC}
In Fig. \ref{fig:cap_region}, we present long-term average throughput regions in a two-user MAC under single-battery and dual-battery cases for symmetric (Fig. \ref{fig:sym}) and asymmetric (Fig. \ref{fig:asym}) settings. 
In both the settings, the battery capacity and the mean value of the harvested energy remain the same with $B_1=B_2=100$ units and $\mu_1=\mu_2=25$. Further, in the symmetric case, the distribution of the energy arrivals is the same in both the users. However, in the asymmetric case,  the energy arrivals are more bursty in user $2$. 
We assume $r_1=r_2=2$ and $r_1=2, \; r_2=1$ in the symmetric and asymmetric cases, respectively. Hence, we have, $\min_{r_u}(G(r_u))=0.51$ and $\min_{r_u}(G(r_u))=0.72$ in the symmetric and asymmetric cases, respectively. 
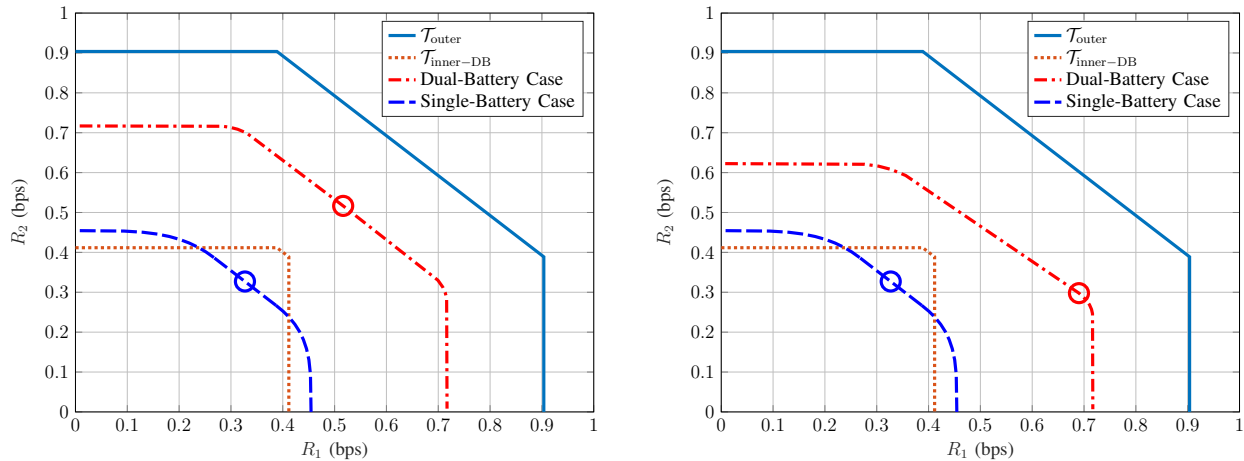
\begin{figure}[t]
	\centering
	\begin{subfigure}{0.48\textwidth}
		% This file was created by matlab2tikz.
%
%The latest updates can be retrieved from
%  http://www.mathworks.com/matlabcentral/fileexchange/22022-matlab2tikz-matlab2tikz
%where you can also make suggestions and rate matlab2tikz.
%
\definecolor{mycolor1}{rgb}{0.00000,0.44700,0.74100}%
\definecolor{mycolor2}{rgb}{0.85000,0.32500,0.09800}%
\begin{tikzpicture}[scale=0.6]
\tikzstyle{densely dashed}=          [dash pattern=on 10pt off 3pt]
\tikzstyle{densely dashdotted}=      [dash pattern=on 8pt off 3pt on \the\pgflinewidth off 3pt]

\begin{axis}[%
width=4.521in,
height=3.482in,
at={(0.758in,0.565in)},
scale only axis,
xmin=0,
xmax=1,
xlabel style={font=\color{white!15!black}},
xlabel={$R_1$ (bps)},
ymin=0,
ymax=1,
ylabel style={font=\color{white!15!black}},
ylabel={$R_2$ (bps)},
axis background/.style={fill=white},
xmajorgrids,
ymajorgrids,
legend style={legend cell align=left, align=left, draw=white!15!black}
]
%%%%%%%%%%%%%%%%%%%%%%%%%%%%%%%%%%%%%%%%%%%%%%%%%%%%%%%%%%%
\addplot [color=mycolor1, line width=2pt]
  table[row sep=crcr]{%
	0	0.903677461028802\\
	0.388803789331776	0.903677461028802\\
	0.903677461028802	0.388803789331776\\
	0.903677461028802	0\\
};
\addlegendentry{$\mathcal{T}_{\rm outer}$}

%%%%%%%%%%%%%%%%%%%%%%%%%%%%%%%%%%%%%%%%%%%%%%%%%%%%%%%%%%%%%%%%
\addplot [color=mycolor2, line width=2pt,dotted]
  table[row sep=crcr]{%
	0	0.411677461028802\\
	0.388803789331776	0.411677461028802\\
	0.411677461028802	0.388803789331776\\
	0.411677461028802	0\\
};
\addlegendentry{$\mathcal{T}_{\rm inner-DB}$}

%%%%%%%%%%%%%%%%%%%%%%%%%%%%%%%%%%%%%%%%%%%%%%%%%%%%%%%%%%%%%%%%%%%%%%%%%%
\addplot [color=red, line width=2pt, draw=none, mark=o, mark options={ red}, mark size=6pt, forget plot]
  table[row sep=crcr]{%
	0.516540075973236	0.516540072752745\\
};

\addplot [color=red, line width=2pt,densely dashdotted]
  table[row sep=crcr]{%
	0.00819528045533224	0.71693951062692\\
	0.283880440605426	0.716157694357974\\
	0.297894183367373	0.713711056615041\\
	0.312795519779515	0.708748105954429\\
	0.330080761155478	0.699379034679675\\
	0.516540075973236	0.516540072752745\\
	0.69937903467963	0.330080761155559\\
	0.708748105954429	0.312795519779508\\
	0.713711056614905	0.297894183367929\\
	0.716157694357948	0.283880440605692\\
	0.716939510626926	0.00819528029836902\\
};
\addlegendentry{Dual-Battery Case}

%%%%%%%%%%%%%%%%%%%%%%%%%%%%%%%%%%%%%%%%%%%%%%%%%%%%%%%%%%%%%%%%%%
\addplot [color=blue, densely dashed, line width=2pt]
  table[row sep=crcr]{%
	0.3869611450239	0.267049277092011\\
	0.400946018574988	0.251948992870873\\
	0.413010654852979	0.236254055743027\\
	0.42340185717397	0.21959005802245\\
	0.432282533647377	0.201499658285084\\
	0.43974304703896	0.181366015664567\\
	0.445804226799671	0.158249333148556\\
	0.450412652748096	0.130438424650707\\
	0.453425485816623	0.0936624568518499\\
	0.454543565170295	0\\
};
\addlegendentry{Single-Battery Case}

\addplot [color=blue, densely dashed, line width=2pt, forget plot]
  table[row sep=crcr]{%
	0.267049277092011	0.3869611450239\\
	0.251948992870873	0.400946018574988\\
	0.236254055743027	0.413010654852979\\
	0.21959005802245	0.42340185717397\\
	0.201499658285084	0.432282533647377\\
	0.181366015664567	0.43974304703896\\
	0.158249333148556	0.445804226799671\\
	0.130438424650707	0.450412652748096\\
	0.0936624568518499	0.453425485816623\\
	0	0.454543565170295\\
};
\addplot [color=blue, densely dashed, line width=2pt, forget plot]
  table[row sep=crcr]{%
	0.3869611450239	0.267049277092011\\
	0.267049277092011	0.3869611450239\\
};

\addplot [color=blue, line width=2pt, draw=none, mark=o, mark options={ blue}, mark size=6pt, forget plot]
  table[row sep=crcr]{%
	0.327005211057955	0.327005211057955\\
};

\end{axis}
\end{tikzpicture}%
		\caption{Symmetric case with  $E_{H_1}=10$ units, $E_{H_2}=10$ units, $p_1=0.25$ and $p_2=0.25$.}
		\label{fig:sym}
	\end{subfigure} 
	\hfill
	\begin{subfigure}{0.48\textwidth}
		% This file was created by matlab2tikz.
%
%The latest updates can be retrieved from
%  http://www.mathworks.com/matlabcentral/fileexchange/22022-matlab2tikz-matlab2tikz
%where you can also make suggestions and rate matlab2tikz.
%
\definecolor{mycolor1}{rgb}{0.00000,0.44700,0.74100}%
\definecolor{mycolor2}{rgb}{0.85000,0.32500,0.09800}%
\begin{tikzpicture}[scale=0.6]
\tikzstyle{densely dashed}=          [dash pattern=on 10pt off 3pt]
\tikzstyle{densely dashdotted}=      [dash pattern=on 8pt off 3pt on \the\pgflinewidth off 3pt]

\begin{axis}[%
width=4.521in,
height=3.482in,
at={(0.758in,0.565in)},
scale only axis,
xmin=0,
xmax=1,
xlabel style={font=\color{white!15!black}},
xlabel={$R_1$ (bps)},
ymin=0,
ymax=1,
ylabel style={font=\color{white!15!black}},
ylabel={$R_2$ (bps)},
axis background/.style={fill=white},
xmajorgrids,
ymajorgrids,
legend style={legend cell align=left, align=left, draw=white!15!black}
]
%%%%%%%%%%%%%%%%%%%%%%%%%%%%%%%%%%%%%%%%%%%%%%%%%%%%%%%%%%%%%%%%%%%
\addplot [color=mycolor1, line width=2pt]
  table[row sep=crcr]{%
	0	0.903677461028802\\
	0.388803789331776	0.903677461028802\\
	0.903677461028802	0.388803789331776\\
	0.903677461028802	0\\
};
\addlegendentry{$\mathcal{T}_{\rm outer}$}

%%%%%%%%%%%%%%%%%%%%%%%%%%%%%%%%%%%%%%%%%%%%%%%%%%%%%%%%%%%%%%%%%%%%%%
\addplot [color=mycolor2, line width=2pt,dotted]
  table[row sep=crcr]{%
	0	0.411677461028802\\
	0.388803789331776	0.411677461028802\\
	0.411677461028802	0.388803789331776\\
	0.411677461028802	0\\
};
\addlegendentry{$\mathcal{T}_{\rm inner-DB}$}

%%%%%%%%%%%%%%%%%%%%%%%%%%%%%%%%%%%%%%%%%%%%%%%%%%%%%%%%%%%%%%%%%%%%%
\addplot [color=red, line width=2pt, draw=none, mark=o, mark options={ red}, mark size=6pt, forget plot]
  table[row sep=crcr]{%
	0.69022028643678	0.297750940716914\\
};

\addplot [color=red, line width=2pt,densely dashdotted]
  table[row sep=crcr]{%
	0.00763893332916099	0.622509577971591\\
	0.276812149602926	0.620867275121238\\
	0.303456476013205	0.616150133654597\\
	0.328011576218479	0.607878736307158\\
	0.355833122145159	0.592810147938802\\
	0.69022028643678	0.297750940716914\\
	0.703352844143848	0.281885821299399\\
	0.710536012804545	0.268730440832015\\
	0.714446697530986	0.257125537109764\\
	0.716347411074435	0.246288384297405\\
	0.716939511071623	0.00677983513652691\\
};
\addlegendentry{Dual-Battery Case}

%%%%%%%%%%%%%%%%%%%%%%%%%%%%%%%%%%%%%%%%%%%%%%%%%%%%%%%%%%%%%%%%%%%%%
\addplot [color=blue, densely dashed, line width=2pt]
  table[row sep=crcr]{%
	0.3869611450239	0.267049277092011\\
	0.400946018574988	0.251948992870873\\
	0.413010654852979	0.236254055743027\\
	0.42340185717397	0.21959005802245\\
	0.432282533647377	0.201499658285084\\
	0.43974304703896	0.181366015664567\\
	0.445804226799671	0.158249333148556\\
	0.450412652748096	0.130438424650707\\
	0.453425485816623	0.0936624568518499\\
	0.454543565170295	0\\
};
\addlegendentry{Single-Battery Case}
\addplot [color=blue, densely dashed, line width=2pt, forget plot]
  table[row sep=crcr]{%
	0.267049277092011	0.3869611450239\\
	0.251948992870873	0.400946018574988\\
	0.236254055743027	0.413010654852979\\
	0.21959005802245	0.42340185717397\\
	0.201499658285084	0.432282533647377\\
	0.181366015664567	0.43974304703896\\
	0.158249333148556	0.445804226799671\\
	0.130438424650707	0.450412652748096\\
	0.0936624568518499	0.453425485816623\\
	0	0.454543565170295\\
};
\addplot [color=blue, densely dashed, line width=2pt, forget plot]
  table[row sep=crcr]{%
	0.3869611450239	0.267049277092011\\
	0.267049277092011	0.3869611450239\\
};
\addplot [color=blue, line width=2pt, draw=none, mark=o, mark options={ blue}, mark size=6pt, forget plot]
  table[row sep=crcr]{%
	0.327005211057955	0.327005211057955\\
};

\end{axis}
\end{tikzpicture}%
		\caption{Asymmetric case with $E_{H_1}=10$ units, $E_{H_2}=20$ units, $p_1=0.25$ and $p_2=0.125$.}
		\label{fig:asym}
	\end{subfigure}	
	\caption{Long-term average throughput regions in a two-user MAC with single and dual batteries with $B_1,B_2=20$ units and $\mu_1,\mu_2=2.5$ in symmetric and asymmetric cases. Markers show the maximum sum throughput points.}
	\label{fig:cap_region}
\end{figure}

In both the figures, as expected, we note that the largest average throughput region in the two-battery case is significantly larger than that in the single-battery case. Moreover, we note that the maximum sum-throughput occurs in the single-battery case at $(0.33, 0.33)$ in both symmetric and asymmetric cases. However, in the two-battery case, the maximum sum throughput occurs at $(0.52,0.52)$ and $(0.69, 0.29)$ in the symmetric and asymmetric cases, respectively. Since  both users have the same individual throughputs at the  maximum sum-throughput point in the single-battery case,  it is more fair than the two-battery case wherein the individual data throughputs are significantly different at its maximum sum-throughput point. This is because, the optimal power allocation in the single-battery case depends only on the mean harvested energy values in both users, unlike in the dual-battery case, wherein the optimal power allocations depend on the distributions of the harvested energy in both users.

{In Fig. \ref{fig:NoU}, we study variation of long-term average sum-throughput, the sum of long-term average throughputs at each user, with the number of users, $U$, when the parameters at all the users are identical with $p=0.2$ and $E_H=10$ (see Fig. \ref{fig:NoU1}) and  $p=0.8$ and $E_H=2.5$ (see Fig. \ref{fig:NoU2}). 
From the figures, we note that the sum-throughputs are concave function of the number of users, as in the case of ideal MAC. Moreover, the performance trends are similar to those in the P2P channel. Specifically, as the capacity of the battery, $B$, and $p$ are increased for a fixed average energy harvesting rate, $\mu$, the performance increases in the dual-battery case, but it remains unchanged in the single-battery case.}
\begin{figure}[t]
	\centering
	\begin{subfigure}{0.48\textwidth}
		% This file was created by matlab2tikz.
%
%The latest updates can be retrieved from
%  http://www.mathworks.com/matlabcentral/fileexchange/22022-matlab2tikz-matlab2tikz
%where you can also make suggestions and rate matlab2tikz.
%
\definecolor{mycolor1}{rgb}{0.00000,0.44700,0.74100}%
\definecolor{mycolor2}{rgb}{0.85000,0.32500,0.09800}%
\definecolor{mycolor3}{rgb}{0.92900,0.69400,0.12500}%
\definecolor{mycolor4}{rgb}{0.49400,0.18400,0.55600}%
\definecolor{mycolor5}{rgb}{0.46600,0.67400,0.18800}%
\begin{tikzpicture}[scale=0.515]

\begin{axis}[%
width=5.425in,
height=3.152in,
at={(0.91in,0.548in)},
scale only axis,
xmin=1,
xmax=8,
xlabel style={font=\bfseries\color{white!15!black}},
xlabel={Number of Users, $U$},
ymin=0,
ymax=2.1,
ylabel style={font=\bfseries\color{white!15!black}},
ylabel={Average Sum-Throughput, $T$ (bps)},
axis background/.style={fill=white},
xmajorgrids,
ymajorgrids,
legend style={at={(0.97,0.03)}, anchor=south east, legend cell align=left, align=left, draw=white!15!black}
]
\addplot [color=mycolor1, line width=1.5pt]
  table[row sep=crcr]{%
1	0.792481250360579\\
2	1.16096404744368\\
3	1.4036774610288\\
4	1.58496250072116\\
5	1.72971580931865\\
6	1.85021985907055\\
7	1.95344529780426\\
8	2.04373142062517\\
};
\addlegendentry{Upper-Bound}

\addplot [color=mycolor3, line width=1.5pt, mark=diamond, mark options={solid, mycolor3},mark size=6pt]
table[row sep=crcr]{%
	1	0.684417480625694\\
	2	1.00342826049522\\
	3	1.21653388939917\\
	4	1.37678997228721\\
	5	1.50526860751845\\
	6	1.61251053840396\\
	7	1.70455053678267\\
	8	1.7852\\
};
\addlegendentry{Dual-Battery, $B=5E_H$}

\addplot [color=mycolor2, line width=1.5pt, mark=square, mark options={solid, mycolor2},mark size=4pt]
  table[row sep=crcr]{%
1	0.559276869119651\\
2	0.823004655146047\\
3	1.00262289677603\\
4	1.13900056906369\\
5	1.24896915784879\\
6	1.34111604067764\\
7	1.42042007293119\\
8	1.49002715483175\\
};
\addlegendentry{Dual-Battery, $B=E_H$}

\addplot [color=mycolor4, line width=1.5pt, mark=o, mark options={solid, mycolor4},mark size=6pt]
  table[row sep=crcr]{%
1	0.401739399921029\\
2	0.607673693833192\\
3	0.751076202672653\\
4	0.862725262989745\\
5	0.954830735431445\\
6	1.03356565717984\\
7	1.10251782303342\\
8	1.16397132181555\\
};
\addlegendentry{Single-Battery, $B=5E_H$}

\addplot [color=mycolor5, line width=1.5pt, mark=asterisk, mark options={solid, mycolor5},mark size=4pt]
  table[row sep=crcr]{%
1	0.401739399921029\\
2	0.607673693833192\\
3	0.751076202672653\\
4	0.862725262989745\\
5	0.954830735431445\\
6	1.03356565717984\\
7	1.10251782303342\\
8	1.16397132181555\\
};
\addlegendentry{Single-Battery, $B=E_H$}

\end{axis}
\end{tikzpicture}%
		\caption{When $\mu=2$ with  $p=0.2$ and $E_H=10$ at each user.}
		\label{fig:NoU1}
	\end{subfigure} 
	\hfill
	\begin{subfigure}{0.48\textwidth}
		% This file was created by matlab2tikz.
%
%The latest updates can be retrieved from
%  http://www.mathworks.com/matlabcentral/fileexchange/22022-matlab2tikz-matlab2tikz
%where you can also make suggestions and rate matlab2tikz.
%
\definecolor{mycolor1}{rgb}{0.00000,0.44700,0.74100}%
\definecolor{mycolor2}{rgb}{0.85000,0.32500,0.09800}%
\definecolor{mycolor3}{rgb}{0.92900,0.69400,0.12500}%
\definecolor{mycolor4}{rgb}{0.49400,0.18400,0.55600}%
\definecolor{mycolor5}{rgb}{0.46600,0.67400,0.18800}%
\begin{tikzpicture}[scale=0.515]

\begin{axis}[%
width=5.425in,
height=3.152in,
at={(0.91in,0.548in)},
scale only axis,
xmin=1,
xmax=8,
xlabel style={font=\bfseries\color{white!15!black}},
xlabel={Number of Users, $U$},
ymin=0,
ymax=2.1,
ylabel style={font=\bfseries\color{white!15!black}},
ylabel={Average Sum-Throughput, $T$ (bps)},
axis background/.style={fill=white},
xmajorgrids,
ymajorgrids,
legend style={at={(0.97,0.03)}, anchor=south east, legend cell align=left, align=left, draw=white!15!black}
]
\addplot [color=mycolor1, line width=1.5pt]
  table[row sep=crcr]{%
1	0.792481250360579\\
2	1.16096404744368\\
3	1.4036774610288\\
4	1.58496250072116\\
5	1.72971580931865\\
6	1.85021985907055\\
7	1.95344529780426\\
8	2.04373142062517\\
};
\addlegendentry{Upper-Bound}

\\
\addplot [color=mycolor3, line width=1.5pt, mark=diamond, mark options={solid, mycolor3},mark size=6pt]
table[row sep=crcr]{%
	1	0.743378205583269\\
	2	1.08735472140327\\
	3	1.31635224776032\\
	4	1.48827199684684\\
	5	1.62596461974635\\
	6	1.74082161415225\\
	7	1.83935073960073\\
	8	1.92562170216093\\
};
\addlegendentry{Dual-Battery, $B=5E_H$}

\addplot [color=mycolor2, line width=1.5pt, mark=square, mark options={solid, mycolor2},mark size=4pt]
  table[row sep=crcr]{%
1	0.722941965521173\\
2	1.03398498811952\\
3	1.23498513516171\\
4	1.38377264695178\\
5	1.50195499201426\\
6	1.59999999970859\\
7	1.68378134564743\\
8	1.75692696809232\\
};
\addlegendentry{Dual-Battery, $B=E_H$}

\addplot [color=mycolor4, line width=1.5pt, mark=o, mark options={solid, mycolor4},mark size=6pt]
  table[row sep=crcr]{%
1	0.401739399921029\\
2	0.607673693833192\\
3	0.751076202672653\\
4	0.862725262989745\\
5	0.954830735431445\\
6	1.03356565717984\\
7	1.10251782303342\\
8	1.16397132181555\\
};
\addlegendentry{Single-Battery, $B=5E_H$}

\addplot [color=mycolor5, line width=1.5pt, mark=asterisk, mark options={solid, mycolor5},mark size=4pt]
  table[row sep=crcr]{%
1	0.401739399921029\\
2	0.607673693833192\\
3	0.751076202672653\\
4	0.862725262989745\\
5	0.954830735431445\\
6	1.03356565717984\\
7	1.10251782303342\\
8	1.16397132181555\\
};
\addlegendentry{Single-Battery, $B=E_H$}

\end{axis}
\end{tikzpicture}%
		\caption{When $\mu=2$ with  $p=0.8$ and $E_H=2.5$ at each user.}
		\label{fig:NoU2}
	\end{subfigure}	
	\caption{Variation of the optimal long-term average sum-throughput with the number of users, $U$.}
	\label{fig:NoU}
\end{figure}
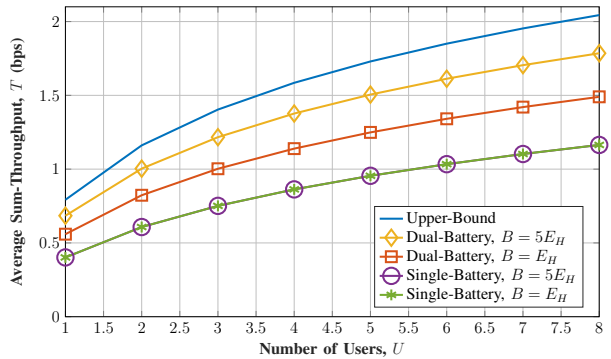
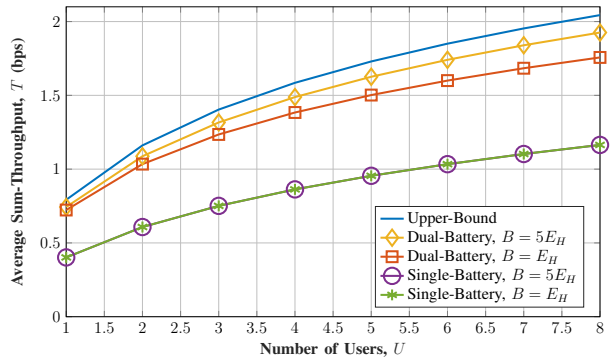

\section{Conclusions}\label{sec:conclusions}
In this work, we optimized the long-term average throughputs and throughput regions in a P2P and multiple-access communication systems under the single-battery and dual-battery cases. In order to avoid capacity degradation of the batteries, we applied the cycle constraint, which says that  a battery must be charged (discharged) only after it is fully discharged (charged). We also applied the  half-duplex constraint for the charging and discharging processes. 
We obtained a closed form expression for the optimal power allocation policy in the P2P channel for the single-battery case. For the dual-battery case,
we first obtained optimal online policy via dynamic programming, followed by non-adaptive and adaptive policies, which do not adapt and adapt power allocations based on battery states, respectively. We then obtained throughput regions in a MAC based on the policies proposed for the P2P channel case. 
We numerically found that the gap between the performance of our suboptimal non-adaptive policy and  the unconstrained optimal performance decreases with the battery capacity faster than the inverse of the square root of the battery capacity, in the dual-battery case. Further,  the optimal throughput in the dual-battery case is significantly higher than that in the single-battery case. 
We also noted that the largest throughput region of the proposed multiple access policy in the single-battery case is contained within that of the dual-battery case. 
%However, the sum-throughput optimal policy in the single-battery case is more fair across the users than that in the dual-battery case. 
%We also obtain long-term average throughput regions for the single-battery and dual-battery cases.  

\section*{Appendix}
\subsection{Proof of Theorem \ref{thm:SB}}
Obtaining the first order differential of the objective function in \eqref{eq:SB_power} and equating it to zero, the optimal $\tilde{P}$ must satisfy, 
\begin{align}\label{eq:lambertW_sol}
\mu+\tilde{P}-(1+\tilde{P})\log(1+\tilde{P})=0.
\end{align}
For all $\mu\geq 1$, \eqref{eq:lambertW_sol} has only one real solution given by $({\mu-1})/({W_0(\exp(-1)(\mu - 1))})-1$, where $W_0(\cdot)$ is the principal branch of the Lambert W function \cite{lambertW}. For $0<\mu<1$, \eqref{eq:lambertW_sol} has two real solutions, namely, $({\mu-1})/({W_0(\exp(-1)(\mu - 1))})-1$ and $({\mu-1})/({W_{-1}(\exp(-1)(\mu - 1))})-1$, where $W_{-1}$ is the lower branch of the Lambert W function. Since the former solution is always larger than the latter solution \cite{lambertW}, we can always choose the larger value and obtain a higher rate without violating any constraints. Hence, the optimal transmit power is given by \eqref{eq:SB opt_tx_power}. Now, noting that $\exp(W_0(x))=x/W_0(x)$, we get   $({\mu-1})/({W_0(\exp(-1)(\mu - 1))})-1=\exp(1)\exp\left(W_0\left(\exp(-1)(\mu-1)\right)\right)-1$. 
Now, we can obtain $\tilde{T}_{\rm SB}$ by substitution of the optimal $\tilde{P}$ in the objective  function.

\subsection{Proof of Theorem  \ref{thm:opt_sol}}
The Lagrangian of \eqref{eq:P1} is given by
\begin{align}
L=&-\sum_{i=1}^{\infty}\bar{F}_{i-1}(r,p)\frac{p}{2r}\log(1+\tilde{P}_i)+\lambda\left(\sum_{i=1}^{\infty}\tilde{P}_i-B\right)-\sum_{i=1}^{\infty}\mu_i\tilde{P}_i, 
\end{align}
where $\lambda,\mu_i\geq 0$ are Lagrange multipliers. 
Differentiating $L$ with respect to $\tilde{P}_i$, we have the following stationarity condition. 
\begin{align}\label{eq:stat}
\frac{\partial L}{\partial \tilde{P}_i}=-\frac{ p\bar{F}_{i-1}(r,p)}{2r\ln 2(1+\tilde{P}_i)}+\lambda-\mu_i=0,\quad i=1,2,\ldots
\end{align}
Further, the complementary slackness conditions, $\lambda\left(\sum_{i=1}^{\infty}\tilde{P}_i-B\right)=0$ and $\mu_i\tilde{P}_i=0$, must be satisfied at the optimal solution. 
Hence, when $\tilde{P}_i>0$, we must have, $\mu_i=0$. From \eqref{eq:stat}, we get, 
\begin{align}\label{eq:power}
\tilde{P}_i=\frac{ p\bar{F}_{i-1}(r,p)}{2r\lambda \ln 2}-1,\quad i=1,2,\ldots, 
\end{align}
where the optimal $\lambda$ can be found from  \eqref{eq:P2_c2}, as shown below.
\paragraph*{Solving for the optimal $\lambda$}
Due to the total energy constraint, we  have, 
\begin{align}\label{eq:sumpower}
B&=\sum_{i=1}^{N}\tilde{P}_i
=\frac{p}{2r\lambda \ln 2}\sum_{i=1}^{N}\bar{F}_{i-1}(r,p)-N\implies \lambda^*=\frac{p\sum_{i=1}^{N}\bar{F}_{i-1}(r,p)}{2r\ln 2(B+N)}, 
\end{align}
where $N$ is the last slot where $\tilde{P}_i>0$. From \eqref{eq:power} and the expression for $\lambda^*$, we find that $N$ is the largest $n$ that satisfies
\begin{align}
{\sum_{i=1}^{n}\bar{F}_{i-1}(r,p)}\leq  ({B+n})F_n(r,p).
\end{align}
Hence, from \eqref{eq:power}, 
\begin{align}\label{eq:tx_power}
\tilde{P}_i=\frac{(B+N)\bar{F}_{i-1}(r,p)}{\sum_{i=1}^{N}\bar{F}_{i-1}(r,p)}-1,\quad i=1,2,\ldots, N
\end{align}
The transmit power in the rest of the slots is zero. Further,  noting that $\bar{F}_{i-1}(r,p)=1$ for $i=1,\ldots,r$, we obtain \eqref{eq:opt_power}.  We note that due to \eqref{eq:sumpower}, the sum of the transmit powers in \eqref{eq:tx_power} never exceeds the battery capacity.

\subsection{Proof of Theorem \ref{thm:approx}}
We now prove the theorem along the lines in Proposition 3 of \cite{Shaviv-Ozgur}. 
\begin{align}\label{eq:additive}
T_{\rm SNA}&=\frac{\mathbb{E}\left[\sum_{i=1}^{L}\frac{1}{2}\log(1+\tilde{P}_i)\right]  }{\mathbb{E}(L)}=\frac{\mathbb{E}\left[\sum_{i=1}^{L}\frac{1}{2}\log(1+\frac{Bp}{r}\sum_{n=i}^{\infty} q_n)\right]  }{\mathbb{E}(L)},\nonumber\\
&\stackrel{\text{a}}{\geq}\frac{\mathbb{E}\left[\sum_{i=1}^{L}\frac{1}{2}\log(1+\frac{Bp}{r})+\sum_{i=1}^{L}\frac{1}{2}\log(\sum_{n=i}^{\infty}q_n)\right]  }{\mathbb{E}(L)},\nonumber\\
&=\frac{\mathbb{E}\left[L\frac{1}{2}\log(1+\frac{Bp}{r})\right]  }{\mathbb{E}(L)}+\frac{\mathbb{E}\left[\sum_{i=1}^{L}\frac{1}{2}\log\left(\sum_{n=i}^{\infty} q_n \right)\right]}{\mathbb{E}(L)},\nonumber\\
%&=\frac{1}{2}\log(1+\frac{Bp}{r})+\frac{p}{r}\sum_{l=1}^{\infty}q_l\sum_{i=1}^{l}\frac{1}{2}\log\left(\sum_{n=i}^{\infty} q_n \right),\nonumber\\
&=\frac{1}{2}\log(1+\frac{Bp}{r})+\frac{p}{r}\sum_{i=1}^{\infty}\left(\sum_{n=i}^{\infty}q_n\right)\frac{1}{2}\log\left( \sum_{n=i}^{\infty} q_n \right)\stackrel{\text{}}{\geq} \frac{1}{2}\log(1+\mu)-G(r), 
\end{align}
where (a) is because $\log(1+\alpha x)\geq \log(1+x)+\log(\alpha)$ for $0\leq \alpha\leq 1$ and we note $0\leq \sum_{n=i}^{\infty}q_n\leq 1$, and  we define  $G(r)\triangleq\max_p (-p/{r})\sum_{i=1}^{\infty}\left(\sum_{n=i}^{\infty}q_n\right)\frac{1}{2}\log\left( \sum_{n=i}^{\infty} q_n \right)$.

\bibliographystyle{ieeetran}
\bibliography{twireless_long}

\end{document}